\documentclass[12pt]{article}
\newcommand{\dsum}{\sum}
\newcommand{\dbigcup}{\bigcup}

\usepackage[T1]{fontenc}
\usepackage{amssymb}
\usepackage{amsmath}
\usepackage{amsfonts}

\newtheorem{theorem}{Theorem}

\newtheorem{corollary}[theorem]{Corollary}

\newtheorem{definition}[theorem]{Definition}

\newtheorem{lemma}[theorem]{Lemma}

\newtheorem{proposition}[theorem]{Proposition}

\newenvironment{proof}[1][Proof]{\noindent\textbf{#1.} }{\ \rule{0.5em}{0.5em}}

\begin{document}

\title{De Donder Construction for Higher Jets}
\author{J\k{e}drzej \'{S}niatycki and Reuven Segev}
\date{}
\maketitle

\section{Introduction}

In 1929, De Donder \cite{de donder 1929}, \cite{de donder 1935}, formulated
an approach to study first order variational problems for several
independent variables in terms of a differential form obtained from the
Lagrangian by the Legendre transformation in all independent variables. His
construction was generalized by Lepage \cite{lepage 1936} yielding a family
of forms, each of which could be used in the same way as the De Donder form
to reduce the original variational problem to a system of equations in
exterior differential forms. A geometric formulation of the De Donder
construction in terms of jets was given by \'{S}niatycki in 1970 \cite%
{sniatycki 1970}. It showed that the De Donder form depended only on the
original Lagrangian and the canonical structure of the appropriate jet
bundle. De Donder form, also called Poincar\'{e}-Cartan form, facilitated an
invariant multisymplectic formulation of field theories \cite%
{kijowski-tulczyjew}, \cite{kijowski-szczyrba}, \cite{BSF}, \cite{gotay}.
The De Donder construction was generalized in 2017 by Kupferman, Olami and
Segev \cite{kupferman} in the context of continuum mechanics of first order
materials, to forms on the first jet bundle that need not be exact. \medskip 

In 1977 Aldaya and Azc\'{a}rraga \cite{aldaya-azcarraga} investigated
generalization of Lepagean forms to higher order variational problems. For
higher order Lagrangians, the natural generalization of the De Donder
construction in terms of Ostrogradski's\footnote{%
English transcription of the original Russian name is Ostrogradsky. However,
Ostrogradski wrote in French and used the French transcription of his name.}
Legendre transformation \cite{ostrogradski}  in all independent variables
leads to a form that depends on the adapted coordinate system used for its
construction. This has lead to search for additional geometric structures,
which would ensure global existence of Poincar\'{e}-Cartan forms, see \cite%
{campos} and references quoted there. In the context of continuum mechanics,
the analysis for higher order jets has been replaced by an analogous, yet
further underdetermined, analysis on iterated jet bundles \cite{segev}%
.\medskip 

In this paper, we generalize De Donder approach to construct boundary forms
that depend on the adapted coordinate system used in the construction. In
continuum mechanics, use of boundary forms leads to splitting of the total
force acting on the body into body force and surface traction. Moreover,
this splitting is independent of the choice of the boundary form used. In
calculus of variations, use of boundary forms leads to equations in exterior
differential forms that are equivalent to the Euler-Lagrange equations.
Infinitesimal symmetries of the theory lead to conservation laws valid for
any choice of the boundary form. In an example, we show that the boundary
conditions lead to independence of constants of motion of the choice of the
boundary form. \medskip 

\noindent \textbf{Acknowledgments} \ The authors are grateful to BIRS for
sponsoring the Banff Workshop on \textquotedblleft Material
Evolution\textquotedblright , June 11-18, 2017, which led to this
collaboration. The first author (J.\'{S}) greatly appreciates constructive
discussions with Mark Gotay and O\u{g}ul Esen.\medskip

\section{Spaces of smooth sections}

We are interested in geometric structure of calculus of variations with $m>1$
independent variables and $n$ dependent variables and its relation to
continuum mechanics. Both theories deal with differentiation of functions on
spaces of maps. There are several approaches to manifold structure of a
space\ of maps. Here, we use the traditional approach of the calculus of
variations flavoured by the insight from theory of differential spaces \cite%
{sniatycki 2013}.\smallskip

Consider a locally trivial fibration $\pi :N\rightarrow M$ with $\dim M=m$
and $\dim N=m+n$. Let $K$ be an open relatively compact sub-manifold of $M$
with smooth boundary $\partial K$. The closure $\bar{K}=K\cup \partial K$ is
a manifold with boundary. A map $\sigma :\bar{K}\rightarrow N$ is a section
of $\pi $ if $\pi \circ \sigma $ is the identity on $\bar{K}$. In the spirit
of theory of differential spaces, we say that a section $\sigma :\bar{K}%
\rightarrow N$ is smooth if it extends to a smooth section of $\pi $ defined
on an open subset of $M$ that contains $\bar{K}.$ We denote by $S^{\infty }(%
\bar{K},N)$ the space of smooth sections $\sigma :\bar{K}\rightarrow N$ of $%
\pi $. In the following, we assume that $\bar{K}$ is contain in the domain
of a chart on $M$. \smallskip 

The next stage is to identify smooth functions on $S^{\infty }(\overline{K}%
,N).$\ We use terminology of jet bundles reviewed in the Appendix. Let $%
\Lambda $ be an $m$-form on the space $J^{\mathrm{k}}(M,N)$ of $\mathrm{k}$%
-jets of sections of $\pi .$ We say that $\Lambda $ is semi-basic with
respect to the source map $\pi ^{\mathrm{k}}:$ $J^{\mathrm{k}%
}(M,N)\rightarrow M$ if $X%
%TCIMACRO{%
%\TeXButton{lefthook}{{\mbox{$ \rule {5pt} {.5pt}\rule {.5pt} {6pt} \, $}}}}%
%BeginExpansion
{\mbox{$ \rule {5pt} {.5pt}\rule {.5pt} {6pt} \, $}}%
%EndExpansion
\Lambda =0$ for every vector field $X$ tangent to fibres of $\pi ^{\mathrm{k}%
}:$ $J^{\mathrm{k}}(M,N)\rightarrow M$, where $%
%TCIMACRO{%
%\TeXButton{lefthook}{{\mbox{$ \rule {5pt} {.5pt}\rule {.5pt} {6pt} \, $}}}}%
%BeginExpansion
{\mbox{$ \rule {5pt} {.5pt}\rule {.5pt} {6pt} \, $}}%
%EndExpansion
\ $denotes the left interior product (contraction) of vectors and forms. The
form $\Lambda $ gives rise to the corresponding action functional 
\begin{equation}
A:S^{\infty }(\bar{K},N)\mapsto \mathbb{R}:\sigma \mapsto A(\sigma
)=\int_{K}j^{\mathrm{k}}\sigma ^{\ast }\Lambda .  \label{2/1}
\end{equation}%
Calculus of variations is concerned with study of critical points of action
functionals. Let $\mathcal{A}$ denote the space of all action functionals on 
$S^{\infty }(\bar{K},N).$ In other words, a function $F:S^{\infty }(\bar{K}%
,N)\rightarrow \mathbb{R}$ is in $\mathcal{A}$ if there exists an integer $%
\mathrm{k}\ \geq 0$, and an $m$-form $\Lambda $ on $J^{\mathrm{k}}(M,N),$
semi-basic with respect to the source map $\pi ^{\mathrm{k}}:$ $J^{\mathrm{k}%
}(M,N)\rightarrow M$, such that 
\begin{equation}
F(\sigma )=\int_{K}j^{\mathrm{k}}\sigma ^{\ast }\Lambda \text{ \ }\forall 
\text{ \ }\sigma \in S^{\infty }(\bar{K},N)\text{.}  \label{2/1a}
\end{equation}%
Here, for $\mathrm{k}=\mathrm{0}$, we use identifications $J^{0}(M,N)=N$ and 
$\pi ^{0}=\pi $. \smallskip 

The tangent space $T_{\sigma }S^{\infty }(\bar{K},N)$ is the space of smooth
maps $Y_{\sigma }:\bar{K}\rightarrow TN$ such that, for each $x\in \bar{K}$, 
$Y_{\sigma }(x)\in T_{\sigma (x)}N$ is tangent to the fibre $\pi ^{-1}(x)$.
It should be noted, that every $Y_{\sigma }\in T_{\sigma }S^{\infty }(\bar{K}%
,N)$ can be extended to a vector field $Y$ on $N$ tangent to fibres of $\pi $
and such that $Y_{\sigma }(x)=Y(\sigma (x))$ for every $x\in \bar{K}$. For $%
Y_{\sigma }\in T_{\sigma }S^{\infty }(\bar{K},N)$ and $F(\sigma )$ given by
equation (\ref{2/1a}), the derivative of $F$ in direction $Y_{\sigma }$ is 
\begin{equation}
D_{Y_{\sigma }}F=\int_{K}j^{\mathrm{k}}\sigma ^{\ast }(\pounds _{Y^{\mathrm{k%
}}}\Lambda ),  \label{2.1b}
\end{equation}%
where $Y^{\mathrm{k}}$ is the prolongation of an extension of $Y_{\sigma }$
to a vertical vector field $Y$ on $N$. It should be noted that the integral
in (\ref{2.1b}) does not depend on the choice of the extension $Y$ of $%
Y_{\sigma }$.\smallskip 

The next step is to identify vector fields on $S^{\infty }(\bar{K},N)$. On
manifolds, vector fields play two roles: they are global derivations of the
differential structure, and they generate local one-parameter local groups
of diffeomorphisms. On manifolds with singularities, eg. stratified spaces,
global derivations need not generate local diffeomorphisms \cite{sniatycki
2013}. In this paper, we consider only vector fields on $S^{\infty }(\bar{K}%
,N)$ that are generated by global vertical vector fields $Y$ on $N$ as
follows. A vertical vector field $Y$ on $N$ gives rise to a section $%
\boldsymbol{Y}:\bar{K}\rightarrow TS^{\infty }(\bar{K},N)$ such that, $%
\boldsymbol{Y}(\sigma )=Y_{\sigma }$ for every $\sigma \in S^{\infty }(\bar{K%
},N)$. In other words, for every $F$ given by equation (\ref{2/1a}) 
\begin{equation}
(\boldsymbol{Y}F)(\sigma )=D_{Y_{\sigma }}F=\int_{K}j^{\mathrm{k}}\sigma
^{\ast }(\pounds _{Y^{\mathrm{k}}}\Lambda )  \label{2/1c}
\end{equation}%
for every $\sigma \in S^{\infty }(\bar{K},N).$ We denote by $\mathfrak{Y}%
(S^{\infty }(\bar{K},N))$ the space of vector fields on $S^{\infty }(\bar{K}%
,N)$ defined above.\smallskip

Now that we have vectors tangent to $S^{\infty }(\bar{K},N)$, we can
consider forms on $S^{\infty }(\bar{K},N)$. Suppose that $\Phi $ is an $%
(m+1) $-form on $J^{\mathrm{k}}(M,N)$ such that $X%
%TCIMACRO{%
%\TeXButton{lefthook}{{\mbox{$ \rule {5pt} {.5pt}\rule {.5pt} {6pt} \, $}}}}%
%BeginExpansion
{\mbox{$ \rule {5pt} {.5pt}\rule {.5pt} {6pt} \, $}}%
%EndExpansion
\Phi $ is semi-basic with respect to the source map $\pi ^{\mathrm{k}}:$ $%
J^{k}(M,N)\rightarrow M$ for every vector field $X$ on $J^{\mathrm{k}}(M,N)$
tangent to fibres of $\pi ^{\mathrm{k}}:$ $J^{\mathrm{k}}(M,N)\rightarrow M$%
. It gives rise to a 1-form $\boldsymbol{\Phi }$ on $S^{\infty }(\bar{K},N)$
defined as follows. For every $\boldsymbol{Y}\in \mathfrak{Y}(S^{\infty }(%
\bar{K},N))$ and $\sigma \in S^{\infty }(\bar{K},N),$ 
\begin{equation}
\left\langle \boldsymbol{\Phi }\mid \boldsymbol{Y}\right\rangle (\sigma
)=\int_{K}j^{\mathrm{k}}\sigma ^{\ast }\left( Y^{\mathrm{k}}%
%TCIMACRO{%
%\TeXButton{lefthook}{{\mbox{$ \rule {5pt} {.5pt}\rule {.5pt} {6pt} \, $}}}}%
%BeginExpansion
{\mbox{$ \rule {5pt} {.5pt}\rule {.5pt} {6pt} \, $}}%
%EndExpansion
\Phi \right) .  \label{2/1d}
\end{equation}

If $Y_{\sigma }\in T_{\sigma }S^{\infty }(\bar{K},N)$ is the restriction of $%
Y$ to $\sigma $, then the restriction to $j^{\mathrm{k}}\sigma (K)$ of the
prolongation $Y^{\mathrm{k}}$ of $Y$ depends only on $Y_{\sigma }$ and not
on its extension off $\sigma (\bar{K})$. This shows that the 1-form $%
\boldsymbol{\Phi }$ restricts to a linear map $\boldsymbol{\Phi }_{\sigma
}:T_{\sigma }S^{\infty }(\bar{K},N)\rightarrow \mathbb{R}$ such that 
\begin{equation}
\left\langle \boldsymbol{\Phi }_{\sigma }\mid Y_{\sigma }\right\rangle
=\left\langle \boldsymbol{\Phi }\mid \boldsymbol{Y}\right\rangle (\sigma ).
\label{2/2}
\end{equation}%
\smallskip

In applications to continuum mechanics, $\bar{K}$ represents the body
manifold, sections $\sigma \in S^{\infty }(\bar{K},N)$ are configurations of
the body, vectors $Y_{\sigma }\in T_{\sigma }S^{\infty }(\bar{K},N)$ are
virtual displacement fields. The form $\boldsymbol{\Phi }$ may be referred
to as a force functional. \smallskip 

\section{Boundary forms}

Let $\Phi $ be an $(m+1)$-form on $J^{\mathrm{k}}(M,N)$ such that $X%
%TCIMACRO{%
%\TeXButton{lefthook}{{\mbox{$ \rule {5pt} {.5pt}\rule {.5pt} {6pt} \, $}}}}%
%BeginExpansion
{\mbox{$ \rule {5pt} {.5pt}\rule {.5pt} {6pt} \, $}}%
%EndExpansion
\Phi $ is semi-basic with respect to the source map $\pi ^{\mathrm{k}}:J^{%
\mathrm{k}}(M,N)\rightarrow M$ for every vector field $X$ on $J^{\mathrm{k}%
}(M,N)$ tangent to fibres of $\pi ^{\mathrm{k}}:$ $J^{\mathrm{k}%
}(M,N)\rightarrow M$. Let $(x^{i},y^{a},z_{i}^{a},...,z_{i_{1}...i_{k}}^{a})$
be local coordinates coordinates on $J^{\mathrm{k}}(M,N)$. The corresponding
local representation of $\Phi $ is 
\begin{eqnarray}
\Phi &=&\sum_{a=1}^{n}\left( \Phi _{a}\mathrm{d}y^{a}+\sum_{i=1}^{m}\Phi
_{a}^{i}\mathrm{d}z_{i}^{a}+...\sum_{i_{1}\leq ...\leq i_{l}}\Phi
_{a}^{i_{1}...i_{l}}\mathrm{d}z_{i_{1}...i_{l}}^{a}\right) \wedge \mathrm{d}%
_{m}x  \label{Phi} \\
&&+...+\sum_{a=1}^{n}\sum_{i_{1}\leq ...\leq i_{k}}\Phi _{a}^{i_{1}...i_{k}}%
\mathrm{d}z_{i_{1}...i_{k}}^{a}\wedge \mathrm{d}_{m}x.  \notag
\end{eqnarray}%
Note that, for every $l=2,...,\mathrm{k}$, coordinates $%
z_{i_{1},...,i_{l}}^{a}$ are symmetric in indices $i_{1},...,i_{l}$ and so
are $\Phi _{a}^{i_{1},...,i_{l}}$. Therefore, the sum in equation (\ref{Phi}%
) is taken only over the independent components. In the following, we modify
the summation convention by the requirement that the indices $%
i_{1},...,i_{l} $ occuring in $z_{i_{1},...,i_{l}}^{a}$ are taken in a
non-decreasing order. This allows us to rewrite equation (\ref{Phi}) as
follows,%
\begin{equation}
\Phi =\left( \Phi _{a}\mathrm{d}y^{a}+\Phi _{a}^{i}\mathrm{d}%
z_{i}^{a}+...\Phi _{a}^{i_{1}...i_{l}}\mathrm{d}z_{i_{1}...i_{l}}^{a}+...+%
\Phi _{a}^{i_{1}...i_{k}}\mathrm{d}z_{i_{1}...i_{k}}^{a}\right) \wedge 
\mathrm{d}_{m}x.  \label{Phi1}
\end{equation}%
An alternative approach would be use of multi-indices. \smallskip

\begin{theorem}
\label{Theorem 1} There exists locally a smooth $m$-form $\Xi $ on $J^{%
\mathrm{2k-1}}(M,N),$ which satisfies the following conditions.

\begin{enumerate}
\item 
\begin{enumerate}
\item $\Xi $ is semi-basic with respect to the forgetful map $\pi _{\mathrm{%
k-1}}^{\mathrm{2k-1}}:J^{\mathrm{2k-1}}(M,N)\rightarrow J^{\mathrm{k}}(M,N).$
In other words, for any vector field $X$ tangent to fibres of $\pi _{\mathrm{%
k-1}}^{\mathrm{2k-1}}:J^{\mathrm{2k-1}}(M,N)\rightarrow J^{\mathrm{k}}(M,N)$%
, 
\begin{equation*}
X%
%TCIMACRO{%
%\TeXButton{lefthook}{{\mbox{$ \rule {5pt} {.5pt}\rule {.5pt} {6pt} \, $}}}}%
%BeginExpansion
{\mbox{$ \rule {5pt} {.5pt}\rule {.5pt} {6pt} \, $}}%
%EndExpansion
\Xi =0.
\end{equation*}

\item For every vector field $X$ on $J^{\mathrm{2k-1}}(M,N)$ tangent to
fibres of the source map $\pi ^{\mathrm{2k-1}}:J^{\mathrm{2k-1}%
}(M,N)\rightarrow M$ the left interior product $X%
%TCIMACRO{%
%\TeXButton{lefthook}{{\mbox{$ \rule {5pt} {.5pt}\rule {.5pt} {6pt} \, $}}}}%
%BeginExpansion
{\mbox{$ \rule {5pt} {.5pt}\rule {.5pt} {6pt} \, $}}%
%EndExpansion
\Xi $ is semi-basic with respect to the source map. Thus, for every pair $%
X_{1},X_{2}$ of vector fields on $J^{\mathrm{2k-1}}(M,N)$ tangent to fibres
of the source map $\pi ^{\mathrm{2k-1}}:J^{\mathrm{2k-1}}(M,N)\rightarrow
M,$%
\begin{equation*}
X_{2}%
%TCIMACRO{%
%\TeXButton{lefthook}{{\mbox{$ \rule {5pt} {.5pt}\rule {.5pt} {6pt} \, $}}}}%
%BeginExpansion
{\mbox{$ \rule {5pt} {.5pt}\rule {.5pt} {6pt} \, $}}%
%EndExpansion
\left( X_{1}%
%TCIMACRO{%
%\TeXButton{lefthook}{{\mbox{$ \rule {5pt} {.5pt}\rule {.5pt} {6pt} \, $}}}}%
%BeginExpansion
{\mbox{$ \rule {5pt} {.5pt}\rule {.5pt} {6pt} \, $}}%
%EndExpansion
\Xi \right) =0.
\end{equation*}
\end{enumerate}

\item For every section $\sigma $ of $\pi :N\rightarrow M$,%
\begin{equation*}
j^{\mathrm{2k-1}}\sigma ^{\ast }\Xi =0,
\end{equation*}%
where $j^{\mathrm{2k-1}}\sigma ^{\ast }\Xi =\Xi \circ \wedge ^{n}T(j^{%
\mathrm{2k-1}}\sigma )$ is the pull-back of $\Xi $ by $j^{\mathrm{2k-1}%
}\sigma :M\rightarrow J^{\mathrm{2k-1}}(M,N).$

\item For every vector field $X$ on $J^{\mathrm{2k-1}}(M,N)$ tangent to
fibres of the target map $\pi _{0}^{\mathrm{2k-1}}:J^{\mathrm{2k-1}%
}(M,N)\rightarrow N,$ and every section $\sigma $ of $\pi :N\rightarrow M,$%
\begin{equation*}
j^{\mathrm{2k-1}}\sigma ^{\ast }\left( X%
%TCIMACRO{%
%\TeXButton{lefthook}{{\mbox{$ \rule {5pt} {.5pt}\rule {.5pt} {6pt} \, $}}}}%
%BeginExpansion
{\mbox{$ \rule {5pt} {.5pt}\rule {.5pt} {6pt} \, $}}%
%EndExpansion
\left( \pi _{\mathrm{k}}^{\mathrm{2k-1}\ast }\Phi +\mathrm{d}\Xi \right)
\right) =0.\smallskip
\end{equation*}
\end{enumerate}
\end{theorem}

\begin{proof}
The first and the second condition imply that $\Xi $ is a linear combination
of contact forms up to order $\mathrm{k}$ with coefficients given by forms
that are semi-basic with respect to the source map. In local coordinates, 
\begin{eqnarray}
\Xi &=&p_{a}^{i}(\mathrm{d}y^{a}-z_{j}^{a}\mathrm{d}x^{j})\wedge \left( 
\frac{{\small \partial }}{{\small \partial x}^{i}}%
%TCIMACRO{%
%\TeXButton{lefthook}{{\mbox{$ \rule {5pt} {.5pt}\rule {.5pt} {6pt} \, $}}}}%
%BeginExpansion
{\mbox{$ \rule {5pt} {.5pt}\rule {.5pt} {6pt} \, $}}%
%EndExpansion
\mathrm{d}_{m}x\right) \notag\\
&&+p_{a}^{i_{1}i_{2}}(\mathrm{d}z_{i%
\,_{2}}^{a}-z_{i_{2}j}^{a}\mathrm{d}x^{j})\wedge \left( \frac{{\small %
\partial }}{{\small \partial x}^{i_{1}}}%
%TCIMACRO{%
%\TeXButton{lefthook}{{\mbox{$ \rule {5pt} {.5pt}\rule {.5pt} {6pt} \, $}}}}%
%BeginExpansion
{\mbox{$ \rule {5pt} {.5pt}\rule {.5pt} {6pt} \, $}}%
%EndExpansion
\mathrm{d}_{m}x\right)   \notag\label{2} \\
&&+...+p_{a}^{i_{1}i_{2}...i_{k}}(\mathrm{d}%
z_{i_{2}...i_{k}}^{a}-z_{i_{2}i_{2}...i_{k}j}^{a}\mathrm{d}x^{j})\wedge
\left( \frac{{\small \partial }}{{\small \partial x}^{i_{1}}}%
%TCIMACRO{%
%\TeXButton{lefthook}{{\mbox{$ \rule {5pt} {.5pt}\rule {.5pt} {6pt} \, $}}}}%
%BeginExpansion
{\mbox{$ \rule {5pt} {.5pt}\rule {.5pt} {6pt} \, $}}%
%EndExpansion
\mathrm{d}_{m}x\right) ,  %\notag
\end{eqnarray}%
where the coefficients $p_{a}^{i_{1}i_{2}...i_{l}}$ are symmetric in indices 
$i_{2},...,i_{l}$, for $l=3,...,k$ and the summation is taken over indices
in non-decreasing order. Hence, 
\begin{eqnarray}
\Xi &=&\left( p_{a}^{i_{1}}\mathrm{d}y^{a}+p_{a}^{i_{1}i_{2}}\mathrm{d}%
z_{i_{2}}^{a}+...+p_{a}^{i_{1}i_{2}...i_{k}}\mathrm{d}z_{i_{2}...i_{k}}^{a}%
\right) \wedge \left( \frac{{\small \partial }}{{\small \partial x}^{i_{1}}}%
%TCIMACRO{%
%\TeXButton{lefthook}{{\mbox{$ \rule {5pt} {.5pt}\rule {.5pt} {6pt} \, $}}}}%
%BeginExpansion
{\mbox{$ \rule {5pt} {.5pt}\rule {.5pt} {6pt} \, $}}%
%EndExpansion
\mathrm{d}_{m}x\right)  \label{4} \\
&&-\left(
p_{a}^{i_{1}}z_{i_{1}}^{a}+p_{a}^{i_{1}i_{2}}z_{i_{1}i_{2}}^{a}+...+p_{a}^{i_{1}i_{2}...i_{k}}z_{i_{1}i_{2}...i_{k}}^{a}\right) 
\mathrm{d}_{m}x.  \notag
\end{eqnarray}%
Note that the symmetry of $z_{i_{1}i_{2}...i_{l}}^{a}$ in $i_{1},...,i_{l}$
implies in equation (\ref{4}) that only the fully symmetric parts of $%
p_{a}^{i_{1}i_{2}...i_{l}}$ contribute to the sum in the second
term.\smallskip

In order to use the third condition, we need the exterior differential of $%
\Xi $. Equation (\ref{2}) yields%
\begin{eqnarray}
\mathrm{d}\Xi &=&\mathrm{d}p_{a}^{i}\wedge (\mathrm{d}y^{a}-z_{j}^{a}\mathrm{%
d}x^{j})\wedge \left( \frac{\partial }{\partial x^{i}}%
%TCIMACRO{%
%\TeXButton{lefthook}{{\mbox{$ \rule {5pt} {.5pt}\rule {.5pt} {6pt} \, $}}}}%
%BeginExpansion
{\mbox{$ \rule {5pt} {.5pt}\rule {.5pt} {6pt} \, $}}%
%EndExpansion
\mathrm{d}_{m}x\right) +  \label{6} \\
&&\mathrm{d}p_{a}^{i_{1}i_{2}}\wedge (\mathrm{d}z_{i%
\,_{2}}^{a}-z_{i_{2}j}^{a}\mathrm{d}x^{j})\wedge \left( \frac{\partial }{%
\partial x^{i_{1}}}%
%TCIMACRO{%
%\TeXButton{lefthook}{{\mbox{$ \rule {5pt} {.5pt}\rule {.5pt} {6pt} \, $}}}}%
%BeginExpansion
{\mbox{$ \rule {5pt} {.5pt}\rule {.5pt} {6pt} \, $}}%
%EndExpansion
\mathrm{d}_{m}x\right) +...+  \notag \\
&&+\mathrm{d}p_{a}^{i_{1}i_{2}...i_{k}}\wedge (\mathrm{d}%
z_{i_{2}...i_{k}}^{a}-z_{i_{2}...i_{k}j}^{a}\mathrm{d}x^{j})\wedge \left( 
\frac{\partial }{\partial x^{i_{1}}}%
%TCIMACRO{%
%\TeXButton{lefthook}{{\mbox{$ \rule {5pt} {.5pt}\rule {.5pt} {6pt} \, $}}}}%
%BeginExpansion
{\mbox{$ \rule {5pt} {.5pt}\rule {.5pt} {6pt} \, $}}%
%EndExpansion
\mathrm{d}_{m}x\right) +  \notag \\
&&-p_{a}^{i}\mathrm{d}z_{j}^{A}\wedge \mathrm{d}x^{j}\wedge \left( \frac{%
\partial }{\partial x^{i}}%
%TCIMACRO{%
%\TeXButton{lefthook}{{\mbox{$ \rule {5pt} {.5pt}\rule {.5pt} {6pt} \, $}}}}%
%BeginExpansion
{\mbox{$ \rule {5pt} {.5pt}\rule {.5pt} {6pt} \, $}}%
%EndExpansion
\mathrm{d}_{m}x\right) -  \notag \\
&&-p_{a}^{i_{1}i_{2}}\mathrm{d}z_{i_{2}j}^{a}\wedge \mathrm{d}x^{j}\wedge
\left( \frac{\partial }{\partial x^{i_{1}}}%
%TCIMACRO{%
%\TeXButton{lefthook}{{\mbox{$ \rule {5pt} {.5pt}\rule {.5pt} {6pt} \, $}}}}%
%BeginExpansion
{\mbox{$ \rule {5pt} {.5pt}\rule {.5pt} {6pt} \, $}}%
%EndExpansion
\mathrm{d}_{m}x\right) +...+  \notag \\
&&-p_{a}^{i_{1}i_{2}...i_{k}}\mathrm{d}z_{i_{2}...i_{k}j}^{a}\wedge \mathrm{d%
}x^{j}\wedge \left( \frac{\partial }{\partial x^{i_{1}}}%
%TCIMACRO{%
%\TeXButton{lefthook}{{\mbox{$ \rule {5pt} {.5pt}\rule {.5pt} {6pt} \, $}}}}%
%BeginExpansion
{\mbox{$ \rule {5pt} {.5pt}\rule {.5pt} {6pt} \, $}}%
%EndExpansion
\mathrm{d}_{m}x\right) ,  \notag
\end{eqnarray}%
which can be simplified to 
\begin{eqnarray}
\mathrm{d}\Xi &=&\mathrm{d}p_{a}^{i}\wedge (\mathrm{d}y^{a}-z_{j}^{a}\mathrm{%
d}x^{j})\wedge \left( \frac{{\small \partial }}{{\small \partial x}^{i}}%
%TCIMACRO{%
%\TeXButton{lefthook}{{\mbox{$ \rule {5pt} {.5pt}\rule {.5pt} {6pt} \, $}}}}%
%BeginExpansion
{\mbox{$ \rule {5pt} {.5pt}\rule {.5pt} {6pt} \, $}}%
%EndExpansion
\mathrm{d}_{m}x\right) +  \label{7} \\
&&\mathrm{d}p_{a}^{i_{1}i_{2}}\wedge (\mathrm{d}z_{i%
\,_{2}}^{a}-z_{i_{2}j}^{a}\mathrm{d}x^{j})\wedge \left( \frac{{\small %
\partial }}{{\small \partial x}^{i}}%
%TCIMACRO{%
%\TeXButton{lefthook}{{\mbox{$ \rule {5pt} {.5pt}\rule {.5pt} {6pt} \, $}}}}%
%BeginExpansion
{\mbox{$ \rule {5pt} {.5pt}\rule {.5pt} {6pt} \, $}}%
%EndExpansion
\mathrm{d}_{m}x\right) +...+  \notag \\
&&+\mathrm{d}p_{a}^{i_{1}i_{2}...i_{k}}\wedge (\mathrm{d}%
z_{i_{2}...i_{k}}^{a}-z_{i_{2}...i_{k}j}^{a}\mathrm{d}x^{j})\wedge \left( 
\frac{{\small \partial }}{{\small \partial x}^{i}}%
%TCIMACRO{%
%\TeXButton{lefthook}{{\mbox{$ \rule {5pt} {.5pt}\rule {.5pt} {6pt} \, $}}}}%
%BeginExpansion
{\mbox{$ \rule {5pt} {.5pt}\rule {.5pt} {6pt} \, $}}%
%EndExpansion
\mathrm{d}_{m}x\right) +  \notag \\
&&-\left( p_{a}^{i}\mathrm{d}z_{i}^{a}+p_{a}^{i_{1}i_{2}}\mathrm{d}%
z_{i_{1}i_{2}}^{a}+...+p_{a}^{i_{1}i_{2}...i_{k}}\mathrm{d}%
z_{i_{1}i_{2}...i_{k}}^{a}\right) \wedge \mathrm{d}_{m}x.\smallskip  \notag
\end{eqnarray}

Let $X$ be a vector field on $J^{\mathrm{2k-1}}(M,N)$ tangent to fibres of
of the target map $\pi _{0}^{\mathrm{2k-1}}:J^{2k-1}(M,N)\rightarrow N$. In
local coordinates,%
\begin{equation}
X=X_{i}^{a}\frac{{\small \partial }}{{\small \partial z_{i}^{a}}}%
+...+X_{i_{1}...i_{2k-1}}^{a}\frac{{\small \partial }}{{\small \partial }%
z_{i_{1}...i_{2k-1}}^{a}}.  \label{8}
\end{equation}%
Note that, for $l=2,...,2k-1,$ $X_{i_{1}...i_{l}}^{a}$ is symmetric in the
indices $i_{1},...,i_{l}$. Then, 
\begin{eqnarray}
X%
%TCIMACRO{%
%\TeXButton{lefthook}{{\mbox{$ \rule {5pt} {.5pt}\rule {.5pt} {6pt} \, $}}}}%
%BeginExpansion
{\mbox{$ \rule {5pt} {.5pt}\rule {.5pt} {6pt} \, $}}%
%EndExpansion
\mathrm{d}\Xi &=&(Xp_{a}^{i})(\mathrm{d}y^{a}-z_{j}^{a}\mathrm{d}%
x^{j})\wedge \left( \frac{{\small \partial }}{{\small \partial x}^{i}}%
%TCIMACRO{%
%\TeXButton{lefthook}{{\mbox{$ \rule {5pt} {.5pt}\rule {.5pt} {6pt} \, $}}}}%
%BeginExpansion
{\mbox{$ \rule {5pt} {.5pt}\rule {.5pt} {6pt} \, $}}%
%EndExpansion
\mathrm{d}_{m}x\right) +...+  \label{9} \\
&&+(Xp_{a}^{i_{1}i_{2}...i_{k}})(\mathrm{d}%
z_{i_{2}...i_{k}}^{a}-z_{i_{2}...i_{k}j}^{a}\mathrm{d}x^{j})\wedge \left( 
\frac{{\small \partial }}{{\small \partial x}^{i}}%
%TCIMACRO{%
%\TeXButton{lefthook}{{\mbox{$ \rule {5pt} {.5pt}\rule {.5pt} {6pt} \, $}}}}%
%BeginExpansion
{\mbox{$ \rule {5pt} {.5pt}\rule {.5pt} {6pt} \, $}}%
%EndExpansion
\mathrm{d}_{m}x\right) +  \notag \\
&&-X_{i\,_{2}}^{a}\mathrm{d}p_{a}^{i_{1}i_{2}}\wedge \left( \frac{{\small %
\partial }}{{\small \partial x}^{i}}%
%TCIMACRO{%
%\TeXButton{lefthook}{{\mbox{$ \rule {5pt} {.5pt}\rule {.5pt} {6pt} \, $}}}}%
%BeginExpansion
{\mbox{$ \rule {5pt} {.5pt}\rule {.5pt} {6pt} \, $}}%
%EndExpansion
\mathrm{d}_{m}x\right) -...-  \notag \\
&&-X_{i_{2}...i_{k}}^{a}\mathrm{d}p_{a}^{i_{1}i_{2}...i_{k}}\wedge \left( 
\frac{{\small \partial }}{{\small \partial x}^{i}}%
%TCIMACRO{%
%\TeXButton{lefthook}{{\mbox{$ \rule {5pt} {.5pt}\rule {.5pt} {6pt} \, $}}}}%
%BeginExpansion
{\mbox{$ \rule {5pt} {.5pt}\rule {.5pt} {6pt} \, $}}%
%EndExpansion
\mathrm{d}_{m}x\right) +  \notag \\
&&-\left(
p_{a}^{i}X_{i}^{a}+p_{a}^{i_{1}i_{2}}X_{i_{1}i_{2}}^{a}+...+p_{a}^{i_{1}i_{2}...i_{k}}X_{i_{1}i_{2}...i_{k}}^{a}\right) 
\mathrm{d}_{m}x,  \notag
\end{eqnarray}%
where $(Xp_{a}^{i})$ is the derivation of $p_{a}^{i}$ in direction $X$,
etc.\smallskip

The first two lines of equation (\ref{9}) do not contribute to $%
j^{2k-1}\sigma ^{\ast }(X%
%TCIMACRO{%
%\TeXButton{lefthook}{{\mbox{$ \rule {5pt} {.5pt}\rule {.5pt} {6pt} \, $}}}}%
%BeginExpansion
{\mbox{$ \rule {5pt} {.5pt}\rule {.5pt} {6pt} \, $}}%
%EndExpansion
\mathrm{d}\Xi )$, because they are linear combinations of contact forms.
Hence, equation (\ref{9}) implies that 
\begin{eqnarray}
j^{2k-1}\sigma ^{\ast }\left( X%
%TCIMACRO{%
%\TeXButton{lefthook}{{\mbox{$ \rule {5pt} {.5pt}\rule {.5pt} {6pt} \, $}}}}%
%BeginExpansion
{\mbox{$ \rule {5pt} {.5pt}\rule {.5pt} {6pt} \, $}}%
%EndExpansion
\mathrm{d}\Xi \right) &=&-X_{i_{2}}^{a}P_{a,i_{1}}^{i_{1}i_{2}}\mathrm{d}%
_{m}x-...-X_{i_{2}...i_{k}}^{a}P_{a,i_{1}}^{i_{1}i_{2}...i_{k}}\mathrm{d}%
_{m}x  \label{10} \\
&&-\left(
P_{a}^{i}X_{i}^{a}+P_{a}^{i_{1}i_{2}}X_{i_{1}i_{2}}^{a}+...+P_{a}^{i_{1}i_{2}...i_{k}}X_{i_{1}i_{2}...i_{k}}^{a}\right) 
\mathrm{d}_{m}x,  \notag
\end{eqnarray}%
where 
\begin{equation}
P_{a}^{i_{1}}=j^{2k-1}\sigma ^{\ast }p_{a}^{i_{1}},\text{ }...,\text{ }%
P_{a}^{i_{1}i_{2}...i_{k}}=j^{2k-1}\sigma ^{\ast }p_{a}^{i_{1}i_{2}...i_{k}},
\label{11}
\end{equation}%
and the components of $X$ are evaluated on the range of $j^{2k-1}\sigma $.
Following the symmetry argument leading to equation (\ref{4}), we can
rewrite equation (\ref{10}) in the form 
\begin{eqnarray}
j^{2k-1}\sigma ^{\ast }\left( X%
%TCIMACRO{%
%\TeXButton{lefthook}{{\mbox{$ \rule {5pt} {.5pt}\rule {.5pt} {6pt} \, $}}}}%
%BeginExpansion
{\mbox{$ \rule {5pt} {.5pt}\rule {.5pt} {6pt} \, $}}%
%EndExpansion
\mathrm{d}\Xi \right)
&=&-X_{i_{2}}^{a}P_{a,i_{1}}^{i_{1}i_{2}}d_{n}x-...-X_{i_{2}...i_{k}}^{a}P_{a,i_{1}}^{i_{1}i_{2}...i_{k}}%
\mathrm{d}_{m}x  \label{11a} \\
&&-\left(
P_{a}^{i}X_{i}^{a}+P_{a}^{(i_{1}i_{2})}X_{i_{1}i_{2}}^{a}+...+P_{a}^{(i_{1}i_{2}...i_{k})}X_{i_{1}i_{2}...i_{k}}^{a}\right) 
\mathrm{d}_{m}x,  \notag
\end{eqnarray}%
where $(i_{1},...,i_{l})$ denotes symmetrization in indices $i_{1},...,i_{l}$%
. Therefore, the third condition of the theorem yields 
\begin{eqnarray}
&&0=\Phi _{a}^{i}X_{i}^{a}+...+\Phi
_{a}^{i_{1}...i_{k}}X_{i_{1}...i_{k}}^{a}-X_{i_{2}}^{a}P_{a,i_{1}}^{i_{1}i_{2}}-
\label{12} \\
&&-...-X_{i_{2}...i_{k}}^{a}P_{a,i_{1}}^{i_{1}i_{2}...i_{k}}-\left(
P_{a}^{i}X_{i}^{a}+P_{a}^{(i_{1}i_{2})}X_{i_{1}i_{2}}^{a}+...+P_{a}^{(i_{1}i_{2}...i_{k})}X_{i_{1}i_{2}...i_{k}}^{a}\right) ,
\notag
\end{eqnarray}%
where $\Phi _{a}^{i},...,\Phi _{a}^{i_{1}...i_{k}}$ are evaluated on the
range $j^{k}\sigma $. Equation (\ref{12}) is equivalent to%
\begin{eqnarray}
0 &=&(\Phi
_{a}^{i_{1}...i_{k}}-P_{a}^{(i_{1}i_{2}...i_{k})})X_{i_{1}i_{2}...i_{k}}^{a}+
\label{13} \\
&&(\Phi
_{a}^{i_{2}...i_{k}}-P_{a}^{(i_{2}...i_{k})}-P_{a,i_{1}}^{i_{1}i_{2}...i_{k}})X_{i_{2}...i_{k}}^{a}+...+
\notag \\
&&+(\Phi _{a}^{i}-P_{a}^{i}-P_{a,i_{1}}^{i_{1}i})X_{i}^{a}.\smallskip  \notag
\end{eqnarray}

Since the components $X_{i_{1}i_{2}..i_{l}}^{a}$ of $X$ are arbitrary
functions symmetric in the indices $i_{1},...,i_{k}$, it follows that 
\begin{eqnarray}
\Phi _{a}^{i_{1}...i_{k}}-P_{a}^{(i_{1}i_{2}...i_{k})} &=&0,  \label{14} \\
\Phi
_{a}^{i_{2}...i_{k}}-P_{a}^{(i_{2}...i_{k})}-P_{a,i_{1}}^{i_{1}i_{2}...i_{k}} &=&0,
\notag \\
\Phi
_{a}^{i_{l}...i_{k}}-P_{a}^{(i_{l}...i_{k})}-P_{a,i_{l-1}}^{i_{l-1}i_{l}...i_{k}} &=&0,~~%
\text{for~~}l=3,...,k-2,  \notag \\
\Phi _{a}^{i}-P_{a}^{i}-P_{a,i_{1}}^{i_{1}i} &=&0,  \notag
\end{eqnarray}%
This shows that there is no unique local form $\Xi $ satisfying the
conditions of our theorem. In order to prove existence, we are free to
impose an additional condition that the coefficients $%
p_{a}^{i_{1}i_{2}...i_{l}}$.\smallskip\ 

An obvious generalization of the De Donder construction corresponds to an
additional conditions that all $p_{a}^{i_{1}i_{2}...i_{l}}$ are fully
symmetric in all indices $i_{1},...,i_{l}$. With this additional assumption,
equation (\ref{14}) yields 
\begin{eqnarray}
\Phi _{a}^{i_{1}...i_{k}}-P_{a}^{i_{1}i_{2}...i_{k}} &=&0,  \label{14a} \\
\Phi
_{a}^{i_{2}...i_{k}}-P_{a}^{i_{2}...i_{k}}-P_{a,i_{1}}^{i_{1}i_{2}...i_{k}}
&=&0,  \notag \\
\Phi
_{a}^{i_{l}...i_{k}}-P_{a}^{i_{l}...i_{k}}-P_{a,i_{l-1}}^{i_{l-1}i_{l}...i_{k}} &=&0,~~%
\text{for~~}l=3,...,k-2,  \notag \\
\Phi _{a}^{i}-P_{a}^{i}-P_{a,i_{1}}^{i_{1}i} &=&0.  \notag
\end{eqnarray}%
In equation (\ref{14a}) $\Phi $ depends on $j^{k}\sigma (x)$. In local
coordinates, the section $\sigma $ is given by $y^{b}=\sigma ^{b}(x)=\sigma
^{b}(x^{1},...,x^{n})$ for $b=1,..m$, and $j^{k}\sigma (x)$ is has
coordinates 
\begin{equation*}
(x^{i},y^{b},z_{j_{1}}^{b},...,z_{,j_{1}...j_{k}}^{b})=(x^{i},\sigma
^{b}(x),\sigma _{,j_{1}}^{b}(x),...,\sigma _{,j_{1}...j_{k}}^{b}(x)).
\end{equation*}%
Hence, 
\begin{equation}
P_{a}^{i_{1}i_{2}...i_{k}}(x)=\Phi _{a}^{i_{1}i_{2}...i_{k}}(x,\sigma
^{b}(x),\sigma _{,j_{1}}^{b}(x),...,\sigma _{,j_{1}...j_{k}}^{b}(x)).
\label{15}
\end{equation}%
Next, 
\begin{eqnarray}
P_{a}^{i_{2}...i_{k}}(x) &=&\Phi _{a}^{i_{2}...i_{k}}(x,\sigma
^{b}(x),\sigma _{,j_{1}}^{b}(x),...,\sigma
_{,j_{1}...j_{k}}^{b}(x))-P_{a,i_{1}}^{i_{1}i_{2}...i_{k}}(x)  \label{15a} \\
&=&\Phi _{a}^{i_{2}...i_{k}}(x,\sigma ^{b}(x),\sigma
_{,j_{1}}^{b}(x),...,\sigma _{,j_{1}...j_{k}}^{b}(x))+  \notag \\
&&-\frac{\partial }{\partial x^{i_{1}}}\left[ \Phi
_{a}^{i_{1}i_{2}...i_{k}}(x,\sigma ^{b}(x),\sigma
_{,j_{1}}^{b}(x),...,\sigma _{,j_{1}...j_{k}}^{b}(x))\right]  \notag \\
&=&\Phi _{a}^{i_{2}...i_{k}}(x,\sigma ^{b}(x),\sigma
_{,j_{1}}^{b}(x),...,\sigma _{,j_{1}...j_{k}}^{b}(x))+  \notag \\
&&-\frac{{\small \partial \Phi }_{a}^{i_{1}i_{2}...i_{k}}}{{\small \partial x%
}^{i_{1}}}(j^{k}\sigma (x))-\frac{{\small \partial \Phi }%
_{a}^{i_{1}i_{2}...i_{k}}}{{\small \partial y}^{b}}(j^{k}\sigma (x))\frac{%
{\small \partial \sigma }^{b}}{{\small \partial x}^{i_{1}}}(x)+  \notag \\
&&-\frac{\partial \Phi _{a}^{i_{1}i_{2}...i_{k}}}{\partial z_{j_{1}}^{b}}%
(j^{k}\sigma (x))\frac{\partial \sigma _{,j_{1}}^{b}}{\partial x^{i_{1}}}%
(x)-...-\frac{\partial \Phi _{a}^{i_{1}i_{2}...i_{k}}}{\partial
z_{j_{1}...j_{k}}^{b}}(j^{k}\sigma (x))\frac{\partial \sigma
_{,j_{1}...j_{k}}^{b}}{\partial x^{i_{1}}}(x).  \notag
\end{eqnarray}%
Since 
\begin{equation*}
\frac{{\small \partial \sigma }_{,j_{1}...j_{k}}^{b}}{{\small \partial x}%
^{i_{1}}}(x)=\sigma _{,j_{1}...j_{k}i_{1}}^{b}(x),
\end{equation*}%
equation (\ref{15a}) implies that $P_{a}^{i_{2}...i_{k}}(x)$ depends on $%
j^{k+1}\sigma (x).$\smallskip

Continuing, we get a complete solution of equations (\ref{14a}) in the form 
\begin{eqnarray}
&&P_{a}^{i_{l}...i_{k}}(x)=\Phi _{a}^{i_{l}...i_{k}}(x,\sigma ^{b}(x),\sigma
_{,j_{1}}^{b}(x),...,\sigma _{,j_{1}...j_{k}}^{b}(x))+  \label{15b} \\
&&+\dsum_{j=1}^{l-1}(-1)^{j}\frac{{\small \partial }^{j}}{{\small \partial x}%
^{i_{l-j}}{\small ...\partial x}^{i_{l-1}}}\left[ \Phi
_{a}^{i_{l-j}i_{l-j+1}...i_{l-1}i_{l-}...i_{k}}(x,\sigma ^{b}(x),\sigma
_{,j_{1}}^{b}(x),...,\sigma _{,j_{1}...j_{k}}^{b}(x))\right]  \notag
\end{eqnarray}%
for $l=1,...,k$. It shows that $P_{a}^{i_{l}...i_{k}}(x)$ depends on $%
j^{k+l-1}\sigma (x)$. In particular, $P_{a}^{i_{k}}(x)$ depends on $%
j^{2k-1}\sigma (x)$.

Recall that $P_{a}^{i_{l}...i_{k}}=j^{2k-1}\sigma ^{\ast
}p_{a}^{i_{l}...i_{k}}$, for $l=1,...,k-1$, where $p_{a}^{i_{l}...i_{k}}$ is
a function on $J^{2k-1}(M,N)$, \thinspace see equation (\ref{11}). Equation (%
\ref{15b}) gives $P_{a}^{i_{l}...i_{k}}$ for all sections $\sigma $ of $\pi $%
. Hence, we can use it to get an explicit expression for $%
p_{a}^{i_{l}...i_{k}}$ as a function of coordinates $%
(x^{i},y^{a},z_{i_{1}}^{a},...z_{i_{1}...i_{2k-1}}^{a}).$ To this end we
define a differential operator 
\begin{equation}
D_{i}=\frac{{\small \partial }}{{\small \partial x}^{i}}+z_{i}^{a}\frac{%
{\small \partial }}{{\small \partial y}^{a}}+z_{ij_{1}}^{a}\frac{{\small %
\partial }}{{\small \partial z}_{j_{1}}^{a}}+...+z_{ij_{1}...j_{2k-2}}^{a}%
\frac{{\small \partial }}{{\small \partial z}_{j_{1}...j_{2k-2}}^{a}}%
+z_{ij_{1}...j_{2k-1}}^{a}\frac{{\small \partial }}{{\small \partial z}%
_{j_{1}...j_{2k-1}}^{a}}  \label{15d}
\end{equation}%
acting on $C^{\infty }(J^{\mathrm{2k-1}}(M,N))$.

It enables us to write%
\begin{eqnarray}
&&p_{a}^{i_{l}...i_{k}}(x^{i},y^{b},z_{j_{1}}^{b},...,z_{j_{1}...j_{2k-1}}^{b})=\Phi _{a}^{i_{l}...i_{k}}(x^{i},y^{b},z_{j_{1}}^{b},...,z_{j_{1}...j_{k}}^{b})+
\label{15e} \\
&&+\dsum_{j=1}^{l-1}(-1)^{j}D_{i_{l-j}}...D_{i_{l-1}}\left[ \Phi
_{a}^{i_{l-j}i_{l-j+1}...i_{l-1}i_{l}...i_{k}}(x^{i},y^{b},z_{j_{1}}^{b},...,z_{j_{1}...j_{k}}^{b})%
\right] .  \notag
\end{eqnarray}%
Hence, in terms of local coordinates $%
(x^{i},y^{b},z_{j_{1}}^{b},...,z_{j_{1}...j_{2k-1}}^{b})$ on $J^{\mathrm{2k-1%
}}(M,N),$ an obvious generalization of the De Donder construction yields
form $\Xi $ given by equation (\ref{2}), where the coefficients $%
p_{a}^{i_{1}},...,p_{a}^{i_{1}i_{2}...i_{k}}$ are given by equation (\ref%
{15e}).\smallskip
\end{proof}

\begin{definition}
\label{Definition 1} Local forms 
\begin{eqnarray*}
\Xi &=&p_{a}^{i}(\mathrm{d}y^{a}-z_{j}^{a}\mathrm{d}x^{j})\wedge \left( 
\frac{{\small \partial }}{{\small \partial x}^{i}}%
%TCIMACRO{%
%\TeXButton{lefthook}{{\mbox{$ \rule {5pt} {.5pt}\rule {.5pt} {6pt} \, $}}}}%
%BeginExpansion
{\mbox{$ \rule {5pt} {.5pt}\rule {.5pt} {6pt} \, $}}%
%EndExpansion
\mathrm{d}_{m}x\right) +p_{a}^{i_{1}i_{2}}(\mathrm{d}z_{i%
\,_{2}}^{a}-z_{i_{2}j}^{a}\mathrm{d}x^{j})\wedge \left( \frac{{\small %
\partial }}{{\small \partial x}^{i_{1}}}%
%TCIMACRO{%
%\TeXButton{lefthook}{{\mbox{$ \rule {5pt} {.5pt}\rule {.5pt} {6pt} \, $}}}}%
%BeginExpansion
{\mbox{$ \rule {5pt} {.5pt}\rule {.5pt} {6pt} \, $}}%
%EndExpansion
\mathrm{d}_{m}x\right) + \\
&&+...+p_{a}^{i_{1}i_{2}...i_{k}}(\mathrm{d}%
z_{i_{2}...i_{k}}^{a}-z_{i_{2}i_{3}...i_{k}j}^{a}\mathrm{d}x^{j})\wedge
\left( \frac{{\small \partial }}{{\small \partial x}^{i_{1}}}%
%TCIMACRO{%
%\TeXButton{lefthook}{{\mbox{$ \rule {5pt} {.5pt}\rule {.5pt} {6pt} \, $}}}}%
%BeginExpansion
{\mbox{$ \rule {5pt} {.5pt}\rule {.5pt} {6pt} \, $}}%
%EndExpansion
\mathrm{d}_{m}x\right) ,
\end{eqnarray*}
are called boundary forms. If $\Xi $ satisfies Condition 3 of Theorem \ref%
{Theorem 1}, we say that $\Xi $ is a \emph{boundary form} of $\Phi $.
\smallskip\ 
\end{definition}

In the following we discuss some properties of boundary forms. This means,
we do not make additional assumptions on the symmetry properties of
coefficients $p_{a}^{i_{1}...i_{l}},$ and do not specify the form $\Phi $
explicitly.\smallskip

\begin{lemma}
\label{Lemma 1} For each vector field $Y$ on $N$, which projects to a vector
field on $M$, and every section $\sigma $ of $\pi :M\rightarrow N$, 
\begin{equation}
j^{\mathrm{2k-1}}\sigma ^{\ast }\left( \pounds _{Y^{2k-1}}\Xi \right) =0%
\text{, }  \label{3.6}
\end{equation}%
where $Y^{\mathrm{2k-1}}$ is the prolongation of $Y$ to $J^{\mathrm{2k-1}%
}(M,N).$
\end{lemma}

\begin{proof}
By definition,%
\begin{equation*}
\pounds _{Y^{\mathrm{2k-1}}}\Xi =\frac{{\small d}}{{\small dt}}(\mathrm{e}%
^{tY^{\mathrm{2k-1}}\ast }\Xi )_{\mid t=0}=\lim_{t\rightarrow 0}\left[ \frac{%
{\small 1}}{{\small t}}\left( \mathrm{e}^{tY^{\mathrm{2k-1}}\ast }\Xi -\Xi
\right) \right] .
\end{equation*}%
Hence, 
\begin{eqnarray*}
j^{\mathrm{2k-1}}\sigma ^{\ast }\left( \pounds _{Y^{\mathrm{2k-1}}}\Xi
\right) &=&j^{\mathrm{2k-1}}\sigma ^{\ast }\lim_{t\rightarrow 0}\left[ \frac{%
{\small 1}}{{\small t}}\left( \mathrm{e}^{tY^{\mathrm{2k-1}}\ast }\Xi -\Xi
\right) \right] \\
&=&\lim_{t\rightarrow 0}\left[ \frac{{\small 1}}{{\small t}}\left( j^{%
\mathrm{2k-1}}\sigma ^{\ast }\mathrm{e}^{tY^{\mathrm{2k-1}}\ast }\Xi -j^{%
\mathrm{2k-1}}\sigma ^{\ast }\Xi \right) \right] \\
&=&\lim_{t\rightarrow 0}\left\{ \frac{{\small 1}}{{\small t}}\left[ (\mathrm{%
e}^{tY^{\mathrm{2k-1}}}\circ j^{\mathrm{2k-1}}\sigma )^{\ast }\Xi -j^{%
\mathrm{2k-1}}\sigma ^{\ast }\Xi \right] \right\} \\
&=&\lim_{t\rightarrow 0}\left\{ \frac{{\small 1}}{{\small t}}\left[ (j^{%
\mathrm{2k-1}}(\mathrm{e}^{tY\ast }\sigma ))^{\ast }\Xi -j^{\mathrm{2k-1}%
}\sigma ^{\ast }\Xi \right] \right\} . \\
&=&\lim_{t\rightarrow 0}\left[ \frac{{\small 1}}{{\small t}}\left( j^{%
\mathrm{2k-1}}(\mathrm{e}^{tY\ast }\sigma )^{\ast }\Xi -j^{\mathrm{2k-1}%
}\sigma ^{\ast }\Xi \right) \right] =0
\end{eqnarray*}%
Condition 2 of Theorem \ref{Theorem 1} ensures that $j^{\mathrm{2k-1}}\sigma
^{\ast }\Xi $ and $(j^{\mathrm{2k-1}}(\mathrm{e}^{tY\ast }\sigma ))^{\ast
}\Xi =0$ for every $t$ in a neighbourhood of $0$. Therefore, $j^{\mathrm{2k-1%
}}\sigma ^{\ast }\left( \pounds _{Y^{\mathrm{2k-1}}}\Xi \right) =0$, which
completes the proof.
\end{proof}

\begin{proposition}
\label{Decomposition} For every $Y_{\sigma }\in T_{\sigma }S^{\infty }(\bar{K%
},N)$, a boundary form $\Xi $ leads to a decomposition 
\begin{equation}
\int_{K}j^{\mathrm{k}}\sigma ^{\ast }\left( Y^{\mathrm{k}}%
%TCIMACRO{%
%\TeXButton{lefthook}{{\mbox{$ \rule {5pt} {.5pt}\rule {.5pt} {6pt} \, $}}}}%
%BeginExpansion
{\mbox{$ \rule {5pt} {.5pt}\rule {.5pt} {6pt} \, $}}%
%EndExpansion
\Phi \right) =\int_{K}\left[ j^{\mathrm{k}}\sigma ^{\ast }(\Phi
_{a})-P_{a,i}^{i}\right] Y_{\sigma }^{a}\mathrm{d}_{m}x+\int_{j^{\mathrm{2k-1%
}}\sigma (\partial K)}(Y_{\sigma }^{\mathrm{2k-1}}%
%TCIMACRO{%
%\TeXButton{lefthook}{{\mbox{$ \rule {5pt} {.5pt}\rule {.5pt} {6pt} \, $}}}}%
%BeginExpansion
{\mbox{$ \rule {5pt} {.5pt}\rule {.5pt} {6pt} \, $}}%
%EndExpansion
\Xi ),  \label{4.8}
\end{equation}%
where $P_{a}^{i}=j^{2k-1}\sigma ^{\ast }p_{a}^{i}$, as in equation (\ref{11}%
).\smallskip 
\end{proposition}

\begin{proof}
Let $Y$ be an extension of $Y_{\sigma }$ to a vertical vector field on $N$.
Clearly, 
\begin{eqnarray*}
\int_{K}j^{\mathrm{k}}\sigma ^{\ast }\left( Y^{k}%
%TCIMACRO{%
%\TeXButton{lefthook}{{\mbox{$ \rule {5pt} {.5pt}\rule {.5pt} {6pt} \, $}}}}%
%BeginExpansion
{\mbox{$ \rule {5pt} {.5pt}\rule {.5pt} {6pt} \, $}}%
%EndExpansion
\Phi \right) &=&\int_{K}j^{\mathrm{2k-1}}\sigma ^{\ast }\left( \pi _{\mathrm{%
k}}^{\mathrm{2k-1}\ast }\left( Y^{\mathrm{k}}%
%TCIMACRO{%
%\TeXButton{lefthook}{{\mbox{$ \rule {5pt} {.5pt}\rule {.5pt} {6pt} \, $}}}}%
%BeginExpansion
{\mbox{$ \rule {5pt} {.5pt}\rule {.5pt} {6pt} \, $}}%
%EndExpansion
\Phi \right) \right) \\
&=&\int_{K}j^{\mathrm{2k-1}}\sigma ^{\ast }\left( Y^{\mathrm{2k-1}}%
%TCIMACRO{%
%\TeXButton{lefthook}{{\mbox{$ \rule {5pt} {.5pt}\rule {.5pt} {6pt} \, $}}}}%
%BeginExpansion
{\mbox{$ \rule {5pt} {.5pt}\rule {.5pt} {6pt} \, $}}%
%EndExpansion
\pi _{\mathrm{k}}^{\mathrm{2k-1}\ast }\Phi \right) \\
&=&\int_{K}j^{\mathrm{2k-1}}\sigma ^{\ast }\left( Y^{\mathrm{2k-1}}%
%TCIMACRO{%
%\TeXButton{lefthook}{{\mbox{$ \rule {5pt} {.5pt}\rule {.5pt} {6pt} \, $}}}}%
%BeginExpansion
{\mbox{$ \rule {5pt} {.5pt}\rule {.5pt} {6pt} \, $}}%
%EndExpansion
\left( \pi _{\mathrm{k}}^{\mathrm{2k-1}\ast }\Phi +\mathrm{d}\Xi -d\Xi
\right) \right) .
\end{eqnarray*}%
Thus,%
\begin{eqnarray}
\int_{K}j^{\mathrm{k}}\sigma ^{\ast }\left( Y^{\mathrm{k}}%
%TCIMACRO{%
%\TeXButton{lefthook}{{\mbox{$ \rule {5pt} {.5pt}\rule {.5pt} {6pt} \, $}}}}%
%BeginExpansion
{\mbox{$ \rule {5pt} {.5pt}\rule {.5pt} {6pt} \, $}}%
%EndExpansion
\Phi \right) &=&\int_{K}j^{\mathrm{2k-1}}\sigma ^{\ast }\left( Y^{\mathrm{%
2k-1}}%
%TCIMACRO{%
%\TeXButton{lefthook}{{\mbox{$ \rule {5pt} {.5pt}\rule {.5pt} {6pt} \, $}}}}%
%BeginExpansion
{\mbox{$ \rule {5pt} {.5pt}\rule {.5pt} {6pt} \, $}}%
%EndExpansion
\left( \pi _{\mathrm{k}}^{\mathrm{2k-1}\ast }\Phi +\mathrm{d}\Xi \right)
\right) +  \label{4.9} \\
&&-\int_{K}j^{\mathrm{2k-1}}\sigma ^{\ast }\left( Y^{\mathrm{2k-1}}%
%TCIMACRO{%
%\TeXButton{lefthook}{{\mbox{$ \rule {5pt} {.5pt}\rule {.5pt} {6pt} \, $}}}}%
%BeginExpansion
{\mbox{$ \rule {5pt} {.5pt}\rule {.5pt} {6pt} \, $}}%
%EndExpansion
d\Xi \right) .  \notag
\end{eqnarray}%
By Condition 3 in Theorem \ref{Theorem 1},%
\begin{equation*}
j^{\mathrm{2k-1}}\sigma ^{\ast }\left( X%
%TCIMACRO{%
%\TeXButton{lefthook}{{\mbox{$ \rule {5pt} {.5pt}\rule {.5pt} {6pt} \, $}}}}%
%BeginExpansion
{\mbox{$ \rule {5pt} {.5pt}\rule {.5pt} {6pt} \, $}}%
%EndExpansion
\left( \pi _{\mathrm{k}}^{\mathrm{2k-1}\ast }\Phi +\mathrm{d}\Xi \right)
\right) =0.\smallskip
\end{equation*}%
for every vector field $X$ tangent to fibres of the target map $\pi _{0}^{%
\mathrm{2k-1}}:J^{\mathrm{2k-1}}(M,N)\rightarrow N$. On the other hand, the
prolongation $Y^{\mathrm{2k-1}}$ of a vertical vector field $Y=Y^{a}\frac{%
\partial }{\partial y^{\alpha }}$ on $N$ is $\pi _{0}^{\mathrm{2k-1}}$%
-related to $Y.$ Therefore, in local coordinates, treating $Y^{a}\frac{%
\partial }{\partial y^{\alpha }}$ as a vector field on $J^{\mathrm{2k-1}%
}(M,N)$, the difference $Y^{\mathrm{2k-1}}-Y^{a}\frac{\partial }{\partial
y^{\alpha }}$ is tangent to fibres of the target map $\pi _{0}^{\mathrm{2k-1}%
}$, so that 
\begin{eqnarray*}
&&j^{\mathrm{2k-1}}\sigma ^{\ast }\left( Y^{\mathrm{2k-1}}%
%TCIMACRO{%
%\TeXButton{lefthook}{{\mbox{$ \rule {5pt} {.5pt}\rule {.5pt} {6pt} \, $}}}}%
%BeginExpansion
{\mbox{$ \rule {5pt} {.5pt}\rule {.5pt} {6pt} \, $}}%
%EndExpansion
\mathrm{d}\left( \pi _{\mathrm{k}}^{\mathrm{2k-1}\ast }\Phi +d\Xi \right)
\right) \\
&=&j^{\mathrm{2k-1}}\sigma ^{\ast }\left( [Y^{\mathrm{2k-1}}-Y^{a}\frac{%
\partial }{\partial y^{\alpha }}+Y^{a}\frac{\partial }{\partial y^{\alpha }}]%
%TCIMACRO{%
%\TeXButton{lefthook}{{\mbox{$ \rule {5pt} {.5pt}\rule {.5pt} {6pt} \, $}}}}%
%BeginExpansion
{\mbox{$ \rule {5pt} {.5pt}\rule {.5pt} {6pt} \, $}}%
%EndExpansion
\mathrm{d}\left( \pi _{\mathrm{k}}^{\mathrm{2k-1}\ast }\Phi +d\Xi \right)
\right) \\
&=&j^{\mathrm{2k-1}}\sigma ^{\ast }\left( Y^{a}\frac{\partial }{\partial
y^{\alpha }}%
%TCIMACRO{%
%\TeXButton{lefthook}{{\mbox{$ \rule {5pt} {.5pt}\rule {.5pt} {6pt} \, $}}}}%
%BeginExpansion
{\mbox{$ \rule {5pt} {.5pt}\rule {.5pt} {6pt} \, $}}%
%EndExpansion
\mathrm{d}\left( \pi _{\mathrm{k}}^{\mathrm{2k-1}\ast }\Phi +\mathrm{d}\Xi
\right) \right) .
\end{eqnarray*}%
Taking into account equations (\ref{7}) and (\ref{11}), we get 
\begin{equation}
j^{\mathrm{2k-1}}\sigma ^{\ast }\left( Y^{\mathrm{2k-1}}%
%TCIMACRO{%
%\TeXButton{lefthook}{{\mbox{$ \rule {5pt} {.5pt}\rule {.5pt} {6pt} \, $}}}}%
%BeginExpansion
{\mbox{$ \rule {5pt} {.5pt}\rule {.5pt} {6pt} \, $}}%
%EndExpansion
\left( \pi _{\mathrm{k}}^{\mathrm{2k-1}\ast }\Phi +\mathrm{d}\Xi \right)
\right) =(j^{\mathrm{k}}\sigma ^{\ast }\Phi _{a}-P_{a,i}^{i})Y_{\sigma }^{a}%
\mathrm{d}_{m}x.  \label{3/9}
\end{equation}

Lemma \ref{Lemma 1} ensures that $\pounds _{Y^{\mathrm{2k-1}}}\Xi =0$.
Hence, 
\begin{equation*}
Y^{\mathrm{2k-1}}%
%TCIMACRO{%
%\TeXButton{lefthook}{{\mbox{$ \rule {5pt} {.5pt}\rule {.5pt} {6pt} \, $}}}}%
%BeginExpansion
{\mbox{$ \rule {5pt} {.5pt}\rule {.5pt} {6pt} \, $}}%
%EndExpansion
d\Xi =-\mathrm{d}(Y^{\mathrm{2k-1}}%
%TCIMACRO{%
%\TeXButton{lefthook}{{\mbox{$ \rule {5pt} {.5pt}\rule {.5pt} {6pt} \, $}}}}%
%BeginExpansion
{\mbox{$ \rule {5pt} {.5pt}\rule {.5pt} {6pt} \, $}}%
%EndExpansion
\Xi ).
\end{equation*}%
Therefore, the second line in equation (\ref{4.9}) can be rewritten in the
form 
\begin{eqnarray*}
-\int_{K}j^{\mathrm{2k-1}}\sigma ^{\ast }\left( Y^{\mathrm{2k-1}}%
%TCIMACRO{%
%\TeXButton{lefthook}{{\mbox{$ \rule {5pt} {.5pt}\rule {.5pt} {6pt} \, $}}}}%
%BeginExpansion
{\mbox{$ \rule {5pt} {.5pt}\rule {.5pt} {6pt} \, $}}%
%EndExpansion
\mathrm{d}\Xi \right) &=&-\int_{K}j^{\mathrm{2k-1}}\sigma ^{\ast }\left( 
\pounds _{Y^{\mathrm{2k-1}}}\Xi -\mathrm{d}(Y^{\mathrm{2k-1}}%
%TCIMACRO{%
%\TeXButton{lefthook}{{\mbox{$ \rule {5pt} {.5pt}\rule {.5pt} {6pt} \, $}}}}%
%BeginExpansion
{\mbox{$ \rule {5pt} {.5pt}\rule {.5pt} {6pt} \, $}}%
%EndExpansion
\Xi )\right) \\
&=&-\int_{K}j^{\mathrm{2k-1}}\sigma ^{\ast }\pounds _{Y^{\mathrm{2k-1}}}\Xi
+\int_{K}j^{\mathrm{2k-1}}\sigma ^{\ast }\mathrm{d}(Y^{\mathrm{2k-1}}%
%TCIMACRO{%
%\TeXButton{lefthook}{{\mbox{$ \rule {5pt} {.5pt}\rule {.5pt} {6pt} \, $}}}}%
%BeginExpansion
{\mbox{$ \rule {5pt} {.5pt}\rule {.5pt} {6pt} \, $}}%
%EndExpansion
\Xi ) \\
&=&\int_{K}\mathrm{d}\left( j^{\mathrm{2k-1}}\sigma ^{\ast }(Y^{\mathrm{2k-1}%
}%
%TCIMACRO{%
%\TeXButton{lefthook}{{\mbox{$ \rule {5pt} {.5pt}\rule {.5pt} {6pt} \, $}}}}%
%BeginExpansion
{\mbox{$ \rule {5pt} {.5pt}\rule {.5pt} {6pt} \, $}}%
%EndExpansion
\Xi )\right) \\
&=&\int_{\partial K}j^{\mathrm{2k-1}}\sigma ^{\ast }(Y^{\mathrm{2k-1}}%
%TCIMACRO{%
%\TeXButton{lefthook}{{\mbox{$ \rule {5pt} {.5pt}\rule {.5pt} {6pt} \, $}}}}%
%BeginExpansion
{\mbox{$ \rule {5pt} {.5pt}\rule {.5pt} {6pt} \, $}}%
%EndExpansion
\Xi ) \\
&=&\int_{j^{2k-1}\sigma (\partial K)}Y_{\sigma }^{\mathrm{2k-1}}%
%TCIMACRO{%
%\TeXButton{lefthook}{{\mbox{$ \rule {5pt} {.5pt}\rule {.5pt} {6pt} \, $}}}}%
%BeginExpansion
{\mbox{$ \rule {5pt} {.5pt}\rule {.5pt} {6pt} \, $}}%
%EndExpansion
\Xi .
\end{eqnarray*}%
This completes the proof.\smallskip
\end{proof}

Next, we want to show that the decomposition (\ref{4.8}) is independent of
the choice of boundary form $\Xi $. Let $X$ be a vector field on $J^{\mathrm{%
2k-1}}(M,N)$ tangent to fibres of the source map $\pi ^{\mathrm{2k-1}}:J^{%
\mathrm{2k-1}}(M,N)\rightarrow M$. For any boundary form $\Xi $ of $\Phi $,
equations (\ref{7}) and (\ref{11}) yield%
\begin{eqnarray*}
&&j^{\mathrm{2k-1}}\sigma ^{\ast }\left( X%
%TCIMACRO{%
%\TeXButton{lefthook}{{\mbox{$ \rule {5pt} {.5pt}\rule {.5pt} {6pt} \, $}}}}%
%BeginExpansion
{\mbox{$ \rule {5pt} {.5pt}\rule {.5pt} {6pt} \, $}}%
%EndExpansion
d\Xi \right) \\
&=&j^{\mathrm{2k-1}}\sigma ^{\ast }\left[ \left( -X^{a}\mathrm{d}%
p_{a}^{i_{1}}-X_{i_{2}}^{a}\mathrm{d}%
p_{a}^{i_{1}i_{2}}-...-X_{i_{2}...i_{k}}^{a}\mathrm{d}%
p_{a}^{i_{1}i_{2}...i_{k}}\right) \wedge \left( \frac{{\small \partial }}{%
{\small \partial x}^{i_{1}}}%
%TCIMACRO{%
%\TeXButton{lefthook}{{\mbox{$ \rule {5pt} {.5pt}\rule {.5pt} {6pt} \, $}}}}%
%BeginExpansion
{\mbox{$ \rule {5pt} {.5pt}\rule {.5pt} {6pt} \, $}}%
%EndExpansion
\mathrm{d}_{m}x\right) \right] \\
&&-j^{\mathrm{2k-1}}\sigma ^{\ast }\left(
p_{a}^{i_{1}}X_{i_{1}}^{a}+p_{a}^{i_{1}i_{2}}X_{i_{1}i_{2}}^{a}+p_{a}^{i_{1}i_{2}...i_{k}}X_{i_{1}i_{2}...i_{k}}^{a}\right) 
\mathrm{d}_{m}x \\
&=&-\left(
X^{a}P_{a,i_{1}}^{i_{1}}+X_{i_{2}}^{a}P_{a,i_{1}}^{i_{1}i_{2}}+...+X_{i_{2}...i_{k}}^{a}P_{a,i_{1}}^{i_{1}i_{2}...i_{k}}\right) 
\mathrm{d}_{m}x \\
&&-\left(
P_{a}^{i_{1}}X_{i_{1}}^{a}+P_{a}^{i_{1}i_{2}}X_{i_{1}i_{2}}^{a}+...+P_{a}^{i_{1}i_{2}...i_{k}}X_{i_{1}i_{2}...i_{k}}^{a}\right) 
\mathrm{d}_{m}x.
\end{eqnarray*}%
Hence 
\begin{eqnarray}
j^{\mathrm{2k-1}}\sigma ^{\ast }\left( X%
%TCIMACRO{%
%\TeXButton{lefthook}{{\mbox{$ \rule {5pt} {.5pt}\rule {.5pt} {6pt} \, $}}}}%
%BeginExpansion
{\mbox{$ \rule {5pt} {.5pt}\rule {.5pt} {6pt} \, $}}%
%EndExpansion
d\Xi \right) &=&-\left[ X^{a}P_{a,i_{1}}^{i_{1}}+X_{i_{2}}^{a}\left(
P_{a,i_{1}}^{i_{1}i_{2}}+P_{a}^{i_{2}}\right) +...\right] \mathrm{d}_{m}x
\label{3/10} \\
&&-\left[ X_{i_{2}...i_{k}}^{a}\left(
P_{a,i_{1}}^{i_{1}i_{2}...i_{k}}+P_{a}^{i_{2}...i_{k}}\right)
+X_{i_{1}i_{2}...i_{k}}^{a}P_{a}^{i_{1}i_{2}...i_{k}}\right] \mathrm{d}_{m}x.
\notag
\end{eqnarray}%
\smallskip Let $\Xi ^{\prime }$ be another boundary form of $\Phi $ such
that 
\begin{eqnarray}
j^{\mathrm{2k-1}}\sigma ^{\ast }\left( X%
%TCIMACRO{%
%\TeXButton{lefthook}{{\mbox{$ \rule {5pt} {.5pt}\rule {.5pt} {6pt} \, $}}}}%
%BeginExpansion
{\mbox{$ \rule {5pt} {.5pt}\rule {.5pt} {6pt} \, $}}%
%EndExpansion
d\Xi ^{\prime }\right) &=&-\left[ X^{a}P_{a,i_{1}}^{\prime
i_{1}}+X_{i_{2}}^{a}\left( P_{a,i_{1}}^{\prime i_{1}i_{2}}+P_{a}^{\prime
i_{2}}\right) +...\right] \mathrm{d}_{m}x  \label{3/11} \\
&&-\left[ X_{i_{2}...i_{k}}^{a}\left( P_{a,i_{1}}^{\prime
i_{1}i_{2}...i_{k}}+P_{a}^{\prime i_{2}...i_{k}}\right)
+X_{i_{1}i_{2}...i_{k}}^{a}P_{a}^{\prime i_{1}i_{2}...i_{k}}\right] \mathrm{d%
}_{m}x,  \notag
\end{eqnarray}%
where the coefficients $P_{a}^{\prime i_{1}i_{2}},...,P_{a}^{\prime
i_{1}i_{2}...i_{k}}$ are symmetric in the upper indices. For the sake of
simplicity, we introduce the notation 
\begin{equation}
Q_{a}^{i_{1}i_{2}...i_{l}}=P_{a}^{i_{1}i_{2}...i_{l}}-P_{a}^{\prime
i_{1}i_{2}...i_{l}}  \label{3/12}
\end{equation}%
for $l=1,...,k.$ Then 
\begin{eqnarray}
j^{\mathrm{2k-1}}\sigma ^{\ast }\left( X%
%TCIMACRO{%
%\TeXButton{lefthook}{{\mbox{$ \rule {5pt} {.5pt}\rule {.5pt} {6pt} \, $}}}}%
%BeginExpansion
{\mbox{$ \rule {5pt} {.5pt}\rule {.5pt} {6pt} \, $}}%
%EndExpansion
d(\Xi -\Xi ^{\prime })\right) &=&-\left[
X^{a}Q_{a,i_{1}}^{i_{1}}+X_{i_{2}}^{a}\left(
Q_{a,i_{1}}^{i_{1}i_{2}}+Q_{a}^{i_{2}}\right) +...\right] \mathrm{d}_{m}x
\label{3/13} \\
&&-\left[ X_{i_{2}...i_{k}}^{a}\left(
Q_{a,i_{1}}^{i_{1}i_{2}...i_{k}}+Q_{a}^{i_{2}...i_{k}}\right)
+X_{i_{1}i_{2}...i_{k}}^{a}Q_{a}^{i_{1}i_{2}...i_{k}}\right] \mathrm{d}_{m}x
\notag
\end{eqnarray}%
Since $\Xi $ and $\Xi ^{\prime }$ are boundary forms of the same form, and $%
X $ is an arbitrary vector field tangent to fibres of the source map,
Condition 3 of Theorem \ref{Theorem 1} yields 
\begin{eqnarray}
Q_{a}^{(i_{1}i_{2}...i_{k})} &=&0,  \label{3/14} \\
Q_{a}^{(i_{2}...i_{k})}+Q_{a,i_{1}}^{i_{1}i_{2}...i_{k}} &=&0,  \notag \\
Q_{a}^{(i_{1}...i_{l})}+Q_{a,i_{1}}^{i_{1}i_{2}...i_{l}} &=&0,~~\text{for~~}%
l=2,...,k,  \notag \\
Q_{a}^{i}+Q_{a,i_{1}}^{i_{1}i_{2}} &=&0,  \notag
\end{eqnarray}%
Note that, by construction, $Q_{a}^{i_{1}i_{2}...i_{l}}$ is symmetric in the
indices $i_{2},...,i_{l}$. Hence, 
\begin{equation}
j^{\mathrm{2k-1}}\sigma ^{\ast }\left( X%
%TCIMACRO{%
%\TeXButton{lefthook}{{\mbox{$ \rule {5pt} {.5pt}\rule {.5pt} {6pt} \, $}}}}%
%BeginExpansion
{\mbox{$ \rule {5pt} {.5pt}\rule {.5pt} {6pt} \, $}}%
%EndExpansion
\mathrm{d}(\Xi -\Xi ^{\prime })\right) =-X^{a}Q_{a,i}^{i}\mathrm{d}_{m}x.
\label{3.15}
\end{equation}

\begin{lemma}
\label{Lemma 2} For boundary forms $\Xi $ and $\Xi ^{\prime }$, given by
equations (\ref{3/10}) and (\ref{3/11}), respectively, 
\begin{equation}
Q_{a,i}^{i}=(P_{a}^{i}-P_{a}^{\prime i})_{,i}=0.  \label{3.16}
\end{equation}%
\smallskip
\end{lemma}

\begin{proof}
We begin with the case when the difference $Q_{a}^{i}=(P_{a}^{i}-P_{a}^{%
\prime i})$ is generated at the highest differential level. In other words,
we consider $%
_{0}Q_{a}^{i_{1}i_{2}...i_{k}}=P_{a}^{i_{1}i_{2}...i_{l}}-P_{a}^{\prime
i_{1}i_{2}...i_{l}}\neq 0$ such that 
\begin{equation}
_{0}Q_{a}^{(i_{1}i_{2}...i_{k})}=0,  \label{3/17}
\end{equation}%
and, the remaining differences are symmetric and satisfy the equations%
\begin{eqnarray}
_{0}Q_{a}^{i_{2}...i_{l}}+~_{0}Q_{a,i_{1}}^{i_{1}i_{2}...i_{l}} &=&0,~~\text{%
for~~}l=3,...,k-1,  \label{3/18} \\
_{0}Q_{a}^{i_{2}}+~_{10}Q_{a,i_{1}}^{i_{1}i_{2}} &=&0.  \notag
\end{eqnarray}%
Therefore, 
\begin{equation}
_{0}Q_{a}^{i_{k}}=(-1)^{k-1}~_{0}Q_{a,i_{1}i_{2}...i_{k-1}}^{i_{1}i_{2}...i_{k-1}i_{k}},
\label{3.19}
\end{equation}%
and 
\begin{eqnarray}
_{0}Q_{a,i_{k}}^{i_{k}}
&=&(-1)^{k-1}~_{0}Q_{a,i_{1}i_{2}...i_{k-1}i_{k}}^{i_{1}i_{2}...i_{k-1}i_{k}}=(-1)^{k-1}~_{0}Q_{a,(i_{1}i_{2}...i_{k-1}i_{k})}^{i_{1}i_{2}...i_{k-1}i_{k}}
\label{3/20} \\
&=&(-1)^{k-1}~_{0}Q_{a,(i_{1}i_{2}...i_{k-1}i_{k})}^{(i_{1}i_{2}...i_{k-1}i_{k})}=(-1)^{k-1}~_{0}Q_{a,i_{1}i_{2}...i_{k-1}i_{k}}^{(i_{1}i_{2}...i_{k-1}i_{k})}=0
\notag
\end{eqnarray}%
because partial derivatives commute.\smallskip

In the next step, we consider the situation when, $\Xi $ and $\Xi ^{\prime }$
agree on the highest differential, that is we assume that $%
_{1}Q_{a}^{i_{1}...i_{k}}=0$. Moreover, we assume that 
\begin{eqnarray}
_{1}Q_{a}^{i_{1}...i_{k-1}} &\neq &0,  \notag \\
_{1}Q_{a}^{(i_{1}...i_{k-1})} &=&0, \\
_{1}Q_{a}^{i_{2}...i_{l}}+~_{1}Q_{a,i_{1}}^{i_{1}i_{2}...i_{l}} &=&0,~~\text{%
for~~}l=3,...,k-1,  \notag \\
_{1}Q_{a}^{i_{2}}+~_{1}Q_{a,i_{1}}^{i_{1}i_{2}} &=&0.  \notag
\end{eqnarray}%
The same arguments as above, lead to 
\begin{equation*}
_{1}Q_{a}^{i_{k}}=(-1)^{k-2}~_{1}Q_{a,i_{1}...i_{k-1}}^{i_{1}...i_{k-1}i_{k}}
\end{equation*}%
so that 
\begin{equation}
_{1}Q_{a,i_{k}}^{i_{k}}=(-1)^{k-2}~_{1}Q_{a,i_{2}...i_{k-1}i_{k}}^{i_{2}...i_{k-1}i_{k}}=0
\label{3/21}
\end{equation}%
because $_{1}Q_{a}^{(i_{2}...i_{k})}=0$. Continuing this procedure, for
every $r=2,,...,k-1$, we consider $_{r}Q_{a}^{i_{l}...i_{k}}$ such that,%
\begin{eqnarray}
_{r}Q_{a}^{i_{1}...i_{k-l}} &=&0\text{, \ for \ }l<r,  \label{3/22} \\
_{r}Q_{a}^{(i_{1}...i_{k-r})} &=&0,  \notag \\
_{r}Q_{a}^{i_{2}...i_{l}}+~_{r}Q_{a,i_{l-1}}^{i_{1}i_{2}...i_{l}} &=&0,~~%
\text{for~~}l=3,...,k-r,  \notag \\
_{r}Q_{a}^{i_{k}}+~_{r}Q_{a,i_{k-1}}^{i_{k-1}i_{k}} &=&0.  \notag
\end{eqnarray}%
As before, for this choice of $_{r}Q_{a}^{i_{1}...i_{l}}$, we have 
\begin{equation}
_{r}Q_{a,i_{k}}^{i_{k}}=(-1)^{k-r}~_{r}Q_{a,i_{r}...i_{k-1}i_{k}}^{i_{2}...i_{k-1}i_{k}}=0.
\label{3/23}
\end{equation}%
\smallskip

The general $Q_{a}^{i_{1}...i_{k}}$ can be expressed as the sum of terms $%
_{r}Q_{a}^{i_{l}...i_{k}}$, for $r=0,...,k-1.$ That is,%
\begin{equation}
Q_{a}^{i_{l}...i_{k}}=~_{0}Q_{a}^{i_{l}...i_{k}}+~_{1}Q_{a}^{i_{l}...i_{k}}+...+~_{k-1}Q_{a}^{i_{l}...i_{k}}.
\label{3/24}
\end{equation}%
The defining equations (\ref{3/22}) for the terms $_{r}Q_{a}^{i_{l}...i_{k}}$
ensure that the decomposition (\ref{3/24}) satisfies equations (\ref{3.16}).
Taking into account equations (\ref{3/20}), (\ref{3/21}) and (\ref{3/23}) we
get%
\begin{equation}
Q_{a,i}^{i}=(~_{0}Q_{a}^{i}+~_{1}Q_{a}^{i}+...+~_{k-1}Q_{a}^{i})_{,i}=~_{0}Q_{a,i}^{i}+~_{1}Q_{a,i}^{i}+...+~_{k-1}Q_{a,i}^{i}=0.
\label{3/25}
\end{equation}

We have shown that $Q_{a,i}^{i}=P_{a,i}^{i}-P_{a,i}^{\prime i}=0$, under the
assumption that $\Xi ^{\prime }$ is the obvious choice of boundary form with
fully symmetric coefficients and no additional assumptions on $\Xi $. Hence,
equation (\ref{3.16}) holds for any pair of boundary forms of the same form $%
\Phi $. We have shown that $P_{a,i}^{\prime i}=P_{a,i}^{i}$ for any other
boundary form $\Xi ^{\prime }$. If $\Xi ^{\prime \prime }$ is still another
boundary form of $\Phi $, then $P_{a,i}^{\prime \prime i}=P_{a,i}^{i}$.
\smallskip
\end{proof}

This implies the following result.

\begin{corollary}
\label{Corollary 1}

\begin{enumerate}
\item If $\Xi $ and $\Xi ^{\prime }$ are boundary forms of the same form $%
\Phi $, then 
\begin{equation}
j^{\mathrm{2k-1}}\sigma ^{\ast }\left( X%
%TCIMACRO{%
%\TeXButton{lefthook}{{\mbox{$ \rule {5pt} {.5pt}\rule {.5pt} {6pt} \, $}}}}%
%BeginExpansion
{\mbox{$ \rule {5pt} {.5pt}\rule {.5pt} {6pt} \, $}}%
%EndExpansion
\mathrm{d}\left( \Xi -\Xi ^{\prime }\right) \right) =0  \label{3/26}
\end{equation}%
for every vector field $X$ tangent to fibres of the source map\newline
$\pi ^{\mathrm{2k-1}}:J^{\mathrm{2k-1}}(M,N)\rightarrow M.$

\item The decomposition (\ref{4.8}) is independent of the choice of boundary
form $\Xi $ for $\Phi $ such that $j^{\mathrm{2k-1}}\sigma (\bar{K})$ is in
the domain of definition of $\Xi $.\smallskip
\end{enumerate}
\end{corollary}

\begin{proof}
Equation (\ref{3/26}) is the consequence of equations (\ref{3.15}) and (\ref%
{3.16})$.$

Equations (\ref{4.8}), (\ref{3.15}) and (\ref{3.16}) yield%
\begin{eqnarray*}
\int_{K}\left[ j^{\mathrm{k}}\sigma ^{\ast }(\Phi _{a})-P_{a,i}^{i}\right]
Y_{\sigma }^{a}\mathrm{d}_{m}x &=&\int_{K}\left[ j^{\mathrm{k}}\sigma ^{\ast
}(\Phi _{a})-P_{a,i}^{\prime i}-(P_{a,i}^{i}-P_{a,i}^{\prime i})\right]
Y_{\sigma }^{a}\mathrm{d}_{m}x \\
&=&\int_{K}\left[ j^{\mathrm{k}}\sigma ^{\ast }(\Phi _{a})-P_{a,i}^{\prime i}%
\right] Y_{\sigma }^{a}\mathrm{d}_{m}x.
\end{eqnarray*}%
because $P_{a,i}^{i}-P_{a,i}^{\prime i}=0.$ Therefore, decompositions (\ref%
{4.8}) for the boundary forms $\Xi $ and $\Xi ^{\prime }$ yield%
\begin{eqnarray*}
\int_{j^{\mathrm{2k-1}}\sigma (\partial K)}(Y_{\sigma }^{\mathrm{2k-1}}%
%TCIMACRO{%
%\TeXButton{lefthook}{{\mbox{$ \rule {5pt} {.5pt}\rule {.5pt} {6pt} \, $}}}}%
%BeginExpansion
{\mbox{$ \rule {5pt} {.5pt}\rule {.5pt} {6pt} \, $}}%
%EndExpansion
\Xi ) &=&\int_{K}j^{\mathrm{k}}\sigma ^{\ast }\left( Y^{\mathrm{k}}%
%TCIMACRO{%
%\TeXButton{lefthook}{{\mbox{$ \rule {5pt} {.5pt}\rule {.5pt} {6pt} \, $}}}}%
%BeginExpansion
{\mbox{$ \rule {5pt} {.5pt}\rule {.5pt} {6pt} \, $}}%
%EndExpansion
\Phi \right) -\int_{K}\left[ j^{\mathrm{k}}\sigma ^{\ast }(\Phi
_{a})-P_{a,i}^{i}\right] Y_{\sigma }^{a}\mathrm{d}_{m}x \\
&=&\int_{K}j^{\mathrm{k}}\sigma ^{\ast }\left( Y^{\mathrm{k}}%
%TCIMACRO{%
%\TeXButton{lefthook}{{\mbox{$ \rule {5pt} {.5pt}\rule {.5pt} {6pt} \, $}}}}%
%BeginExpansion
{\mbox{$ \rule {5pt} {.5pt}\rule {.5pt} {6pt} \, $}}%
%EndExpansion
\Phi \right) -\int_{K}\left[ j^{\mathrm{k}}\sigma ^{\ast }(\Phi
_{a})-P_{a,i}^{\prime i}\right] Y_{\sigma }^{a}\mathrm{d}_{m}x \\
&=&\int_{j^{\mathrm{2k-1}}\sigma (\partial K)}(Y_{\sigma }^{\mathrm{2k-1}}%
%TCIMACRO{%
%\TeXButton{lefthook}{{\mbox{$ \rule {5pt} {.5pt}\rule {.5pt} {6pt} \, $}}}}%
%BeginExpansion
{\mbox{$ \rule {5pt} {.5pt}\rule {.5pt} {6pt} \, $}}%
%EndExpansion
\Xi ^{\prime }).
\end{eqnarray*}%
This shows that the decomposition (\ref{4.8}) is independent of the choice
of $\Xi .$\smallskip
\end{proof}

Since boundary forms are constructed in terms of adapted coordinate systems,
non-uniquennes of the De Donder construction implies only local existence of
the result. We see in Corollary \ref{Corollary 1} that decomposition (\ref%
{4.8}) does not depend on the choice of boundary form with the same domain
of definition. If boundary forms are globally defined, then decomposition (%
\ref{4.8}) is unique and it holds for every section of $\pi $ and each
relatively compact open submanifold $K$ of $M$ with piece-wise smooth
boundary $\partial K$. The existence of global boundary forms is a
topological condition on the fibration $\pi :N\rightarrow M$. It is
satisfied if the fibration is trivial and $M$ and the stypical fibre of $\pi 
$ are diffeomorphic to open subsets of $\mathbb{R}^{m}$ and $\mathbb{R}^{n}$%
, respectively. In particular, it is satisfied in many problems in continuum
mechanics. \smallskip

\section{Application to variational problems}

\subsection{Critical points of action functionals}

In this section, we consider the case when $\Phi =\mathrm{d}\Lambda $, where 
$\Lambda $ is a semi-basic $m$-form on $J^{\mathrm{k}}(M,N)$. Let $%
K\subseteq M$ be an open relatively compact submanifold of $M$ with
piece-wise smooth boundary $\partial K$. As in Section 2, we consider the
space $S^{\infty }(\bar{K},N)$ of smooth section $\sigma :\bar{K}\rightarrow
N.$ The form $\Lambda $ defines an action functional $A$ on $S^{\infty }(%
\bar{K},N)$, given by 
\begin{equation}
A(\sigma )=\int_{K}j^{\mathrm{k}}\sigma ^{\ast }\Lambda =\int_{j^{\mathrm{k}%
}\sigma (K)}\Lambda .  \label{4/1}
\end{equation}

\begin{definition}
A section $\sigma \in S^{\infty }(\bar{K},N)$ is a \emph{critical point} of $%
A$ if $D_{Y_{\sigma }}A=0$ for every $Y_{\sigma }\in T_{\sigma }S^{\infty }(%
\overline{K},N),$ which vanishes on $\partial K$ together with its partial
derivatives up to order $\mathrm{k-1}$.
\end{definition}

Taking into account equation (\ref{2.1b}), we see that $\sigma \in S^{\infty
}(\bar{K},N)$ is a \emph{critical point} of $A$ if 
\begin{equation}
\int_{j^{\mathrm{k}}\sigma (K)}\pounds _{Y^{\mathrm{k}}}\Lambda =0
\label{4/2}
\end{equation}%
for every $Y_{\sigma }\in T_{\sigma }S^{\infty }(\bar{K},N),$ which vanishes
on on $\partial K$ together with its partial derivatives up to order $%
\mathrm{k}-1$. Here, $Y^{\mathrm{k}}$ is the prolongation to $J^{\mathrm{k}%
}(M,N)$ of an extension of $Y_{\sigma }$ to a vertical vector field $Y$ on $N
$.\smallskip 

For every vector field $Y$ on $N$, 
\begin{equation}
\pounds _{Y^{\mathrm{k}}}\Lambda =Y^{\mathrm{k}}%
%TCIMACRO{%
%\TeXButton{lefthook}{{\mbox{$ \rule {5pt} {.5pt}\rule {.5pt} {6pt} \, $}}}}%
%BeginExpansion
{\mbox{$ \rule {5pt} {.5pt}\rule {.5pt} {6pt} \, $}}%
%EndExpansion
\mathrm{d}\Lambda +\mathrm{d}\left( Y^{\mathrm{k}}%
%TCIMACRO{%
%\TeXButton{lefthook}{{\mbox{$ \rule {5pt} {.5pt}\rule {.5pt} {6pt} \, $}}}}%
%BeginExpansion
{\mbox{$ \rule {5pt} {.5pt}\rule {.5pt} {6pt} \, $}}%
%EndExpansion
\Lambda \right) .  \label{4/3}
\end{equation}%
The identity (\ref{4/3}) and Stokes' Theorem yield%
\begin{eqnarray*}
\int_{j^{\mathrm{k}}\sigma (K)}\pounds _{Y^{\mathrm{k}}}\Lambda &=&\int_{j^{%
\mathrm{k}}\sigma (K)}\left[ Y^{\mathrm{k}}%
%TCIMACRO{%
%\TeXButton{lefthook}{{\mbox{$ \rule {5pt} {.5pt}\rule {.5pt} {6pt} \, $}}}}%
%BeginExpansion
{\mbox{$ \rule {5pt} {.5pt}\rule {.5pt} {6pt} \, $}}%
%EndExpansion
\mathrm{d}\Lambda +\mathrm{d}\left( Y^{\mathrm{k}}%
%TCIMACRO{%
%\TeXButton{lefthook}{{\mbox{$ \rule {5pt} {.5pt}\rule {.5pt} {6pt} \, $}}}}%
%BeginExpansion
{\mbox{$ \rule {5pt} {.5pt}\rule {.5pt} {6pt} \, $}}%
%EndExpansion
\Lambda \right) \right] \\
&=&\int_{j^{\mathrm{k}}\sigma (K)}Y^{\mathrm{k}}%
%TCIMACRO{%
%\TeXButton{lefthook}{{\mbox{$ \rule {5pt} {.5pt}\rule {.5pt} {6pt} \, $}}}}%
%BeginExpansion
{\mbox{$ \rule {5pt} {.5pt}\rule {.5pt} {6pt} \, $}}%
%EndExpansion
\mathrm{d}\Lambda +\int_{\partial j^{\mathrm{k}}\sigma (K)}\left( Y^{\mathrm{%
k}}%
%TCIMACRO{%
%\TeXButton{lefthook}{{\mbox{$ \rule {5pt} {.5pt}\rule {.5pt} {6pt} \, $}}}}%
%BeginExpansion
{\mbox{$ \rule {5pt} {.5pt}\rule {.5pt} {6pt} \, $}}%
%EndExpansion
\Lambda \right) \\
&=&\int_{j^{\mathrm{k}}\sigma (K)}Y^{\mathrm{k}}%
%TCIMACRO{%
%\TeXButton{lefthook}{{\mbox{$ \rule {5pt} {.5pt}\rule {.5pt} {6pt} \, $}}}}%
%BeginExpansion
{\mbox{$ \rule {5pt} {.5pt}\rule {.5pt} {6pt} \, $}}%
%EndExpansion
\mathrm{d}\Lambda
\end{eqnarray*}%
because $\Lambda $ is semi-basic with respect to the source map $\pi ^{%
\mathrm{k}}:J^{\mathrm{k}}(M,N)\rightarrow M$. Hence equation (\ref{4/2}) is
equivalent to%
\begin{equation}
\int_{K}j^{\mathrm{k}}\sigma ^{\ast }\left( Y^{\mathrm{k}}%
%TCIMACRO{%
%\TeXButton{lefthook}{{\mbox{$ \rule {5pt} {.5pt}\rule {.5pt} {6pt} \, $}}}}%
%BeginExpansion
{\mbox{$ \rule {5pt} {.5pt}\rule {.5pt} {6pt} \, $}}%
%EndExpansion
\mathrm{d}\Lambda \right) =0  \label{4/5}
\end{equation}%
for every vertical vector field $Y$ on $N$. Therefore, $\sigma $ is a
critical section of $A$ if, equation (\ref{4/5}) holds for every vertical
vector field $Y$ on $N$ such that $Y^{\mathrm{k-1}}$ vanishes on $\partial
j^{\mathrm{k-1}}\sigma (K)=j^{\mathrm{k-1}}\sigma (\partial K)$.\smallskip

\subsection{Euler-Lagrange equations}

The Euler-Lagrange equations are obtained by using the coordinate
description of $\Lambda ,$ 
\begin{equation}
\Lambda =L(x^{i},y^{a},z_{i_{1}}^{a},...,z_{i_{1},...,i_{k}}^{a})\mathrm{d}%
_{m}x.  \label{4/6}
\end{equation}%
The usual rule that "variation of the derivative is the derivative of the
variation" corresponds to the choice of extension of $Y_{\sigma }$ to a
vertical vector field $Y=Y^{a}(x^{i})\frac{{\small \partial }}{{\small %
\partial y}^{a}}$ with components independent of $y^{a}$. Its prolongation
to $J^{\mathrm{k}}(M,N)$ is 
\begin{equation}
Y^{\mathrm{k}}(x^{i},z_{j_{1}}^{b},...,z_{j_{1}...j_{k}}^{b})=Y^{a}(x^{i})%
\frac{{\small \partial }}{{\small \partial y}^{a}}+Y_{,i}^{a}(x^{i})\frac{%
{\small \partial }}{{\small \partial z}_{i}^{a}}%
+...+Y_{,i_{1}...i_{k}}^{a}(x^{i})\frac{{\small \partial }}{{\small \partial
z}_{i_{1}...i_{k}}^{a}}.  \label{4/7}
\end{equation}%
Finally, the coordinate description of $j^{\mathrm{k}}\sigma $ is 
\begin{equation}
j^{\mathrm{k}}\sigma :M\rightarrow J^{\mathrm{k}}(M,N):(x^{i})\mapsto
(x^{i},y^{a}(x),z_{i_{1}}^{a}(x),...,z_{i_{1},...,i_{k}}^{a}(x)),
\label{4/8}
\end{equation}%
where%
\begin{equation}
z_{i_{1},...,i_{l}}^{a}(x)=y^{a}(x)_{,i_{1},...,i_{l}}  \label{4/9}
\end{equation}%
for every positive integer $l$. With this notation,%
\begin{equation}
j^{\mathrm{k}}\sigma ^{\ast }\left( Y^{\mathrm{k}}%
%TCIMACRO{%
%\TeXButton{lefthook}{{\mbox{$ \rule {5pt} {.5pt}\rule {.5pt} {6pt} \, $}}}}%
%BeginExpansion
{\mbox{$ \rule {5pt} {.5pt}\rule {.5pt} {6pt} \, $}}%
%EndExpansion
\mathrm{d}(L\mathrm{d}_{m}x)\right) =\left( \frac{\partial L}{\partial
y^{a}}Y^{a}+\frac{\partial L}{\partial z_{i}^{a}}Y_{,i}^{a}+...+%
\frac{\partial L}{\partial z_{i_{1}...i_{k}}^{a}}Y_{,i_{1}...i_{k}}^{a}%
\right) \mathrm{d}_{m}x,  \label{4/10}
\end{equation}%
where all quantities on the right hand side are expressed as functions of $%
(x^{1},...,x^{n}).$ Integrating this result over $K$ and using Stokes
Theorem yields%
\begin{eqnarray*}
&&\int_{K}j^{\mathrm{k}}\sigma ^{\ast }\left( Y^{\mathrm{k}}%
%TCIMACRO{%
%\TeXButton{lefthook}{{\mbox{$ \rule {5pt} {.5pt}\rule {.5pt} {6pt} \, $}}}}%
%BeginExpansion
{\mbox{$ \rule {5pt} {.5pt}\rule {.5pt} {6pt} \, $}}%
%EndExpansion
\mathrm{d}\Lambda \right) =\int_{K}\left( \frac{\partial L}{\partial
y^{a}}Y^{a}+\frac{\partial L}{\partial z_{i_{1}}^{a}}%
Y_{,i_{1}}^{a}+...+\frac{\partial L}{\partial z_{i_{1}...i_{k}}^{a}}%
Y_{,i_{1}...i_{k}}^{a}\right) \mathrm{d}_{m}x \\
&=&\int_{K}\left[ \frac{\partial L}{\partial y^{a}}Y^{a}+\frac{%
\partial }{\partial x^{i_{1}}}\left( \frac{\partial L}{\partial z_{i_{1}}^{a}%
}Y^{a}\right) -\left( \frac{\partial }{\partial x^{i_{1}}}\frac{\partial L}{%
\partial z_{i_{1}}^{a}}\right) Y^{a}+...\right] \mathrm{d}_{m}x+ \\
&&+\int_{K}\left[ \frac{\partial }{\partial x^{i_{1}}}\left( \frac{\partial L%
}{\partial z_{i_{1}...i_{k}}^{a}}Y_{,i_{2}...i_{k}}^{a}\right) -\left( \frac{%
\partial }{\partial x^{i_{1}}}\frac{\partial L}{\partial
z_{i_{1}...i_{k}}^{a}}\right) Y_{,i_{2}...i_{k}}^{a}\right] \mathrm{d}_{m}x,
\\
&=&\int_{K}\left[ \frac{\partial L}{\partial y^{a}}Y^{a}-\left( \frac{%
\partial }{\partial x^{i}}\frac{\partial L}{\partial z_{i}^{a}}\right)
Y^{a}-...-\left( \frac{\partial }{\partial x^{i_{k}}}\frac{\partial L}{%
\partial z_{i_{1}...i_{k}}^{a}}\right) Y_{,i_{1}...i_{k-1}}^{a}\right] 
\mathrm{d}_{m}x \\
&&+\int_{\partial K}\left[ \frac{\partial L}{\partial z_{i_{1}}^{a}}%
Y^{a}+...+\frac{\partial L}{\partial z_{i_{1}...i_{k}}^{a}}%
Y_{,i_{2}...i_{k}}^{a}\right] \left( \frac{\partial }{\partial x^{i_{1}}}%
%TCIMACRO{%
%\TeXButton{lefthook}{{\mbox{$ \rule {5pt} {.5pt}\rule {.5pt} {6pt} \, $}}}}%
%BeginExpansion
{\mbox{$ \rule {5pt} {.5pt}\rule {.5pt} {6pt} \, $}}%
%EndExpansion
\mathrm{d}_{m}x\right) \\
&=&\int_{K}\left[ \frac{\partial L}{\partial y^{a}}Y^{a}-\left( \frac{%
\partial }{\partial x^{i}}\frac{\partial L}{\partial z_{i}^{a}}\right)
Y^{a}-...-\left( \frac{\partial }{\partial x^{i_{k}}}\frac{\partial L}{%
\partial z_{i_{1}...i_{k}}^{a}}\right) Y_{,i_{1}...i_{k-1}}^{a}\right] 
\mathrm{d}_{m}x
\end{eqnarray*}%
because $Y^{\mathrm{k}-1}$ vanishes on $j^{\mathrm{k-1}}\sigma (\partial K)$%
. Continuing integration by parts, we get%
\begin{eqnarray}
\int_{K}j^{\mathrm{k}}\sigma ^{\ast }\left( Y^{\mathrm{k}}%
%TCIMACRO{%
%\TeXButton{lefthook}{{\mbox{$ \rule {5pt} {.5pt}\rule {.5pt} {6pt} \, $}}}}%
%BeginExpansion
{\mbox{$ \rule {5pt} {.5pt}\rule {.5pt} {6pt} \, $}}%
%EndExpansion
d\Lambda \right) &=&\int_{K}\left( \frac{\partial L}{\partial y^{a}}-\frac{%
\partial }{\partial x^{i_{1}}}\frac{\partial L}{\partial z_{i_{1}}^{a}}%
+...+(-1)^{k}\frac{\partial ^{k}}{\partial x^{i_{1}}...\partial x^{i_{k}}}%
\frac{\partial L}{\partial z_{i_{1}...i_{k}}^{a}}\right) Y^{a}\mathrm{d}_{m}x
\notag \\
&=&\int_{K}\frac{\delta L}{\delta y^{a}}Y^{a}\mathrm{d}_{m}x,  \label{4/11}
\end{eqnarray}%
where 
\begin{equation}
\frac{\delta L}{\delta y^{a}}=\frac{\partial L}{\partial y^{a}}-\frac{%
\partial }{\partial x^{i_{1}}}\frac{\partial L}{\partial z_{i_{1}}^{a}}%
+...+(-1)^{k}\frac{\partial ^{k}}{\partial x^{i_{1}}...\partial x^{i_{k}}}%
\frac{\partial L}{\partial z_{i_{1}...i_{k}}^{a}}  \label{4/11a}
\end{equation}%
is called the \emph{Lagrange derivative} of $L$. Comparing equation (\ref%
{4/11}) with equation (\ref{4.8}) observe that, if $\Phi =d(Ld_{m}x)$, then $%
j^{\mathrm{k}}\sigma ^{\ast }(\Phi _{a})-P_{a,i}^{i}=\frac{\delta L}{\delta
y^{a}}.$\smallskip

Taking into account the Fundamental Theorem in the Calculus of Variations,
we conclude that $\sigma $ is a critical section of $A_{K}$ if and only if,
for every $a=1,...,n$, 
\begin{equation}
\left( \frac{\partial L}{\partial y^{a}}-\frac{\partial }{\partial x^{i_{1}}}%
\frac{\partial L}{\partial z_{i_{1}}^{a}}+...+(-1)^{k}\frac{\partial ^{k}}{%
\partial x^{i_{1}}...\partial x^{i_{k}}}\frac{\partial L}{\partial
z_{i_{1}...i_{k}}^{a}}\right) _{\mid K}=0.  \label{4/12}
\end{equation}%
Equations (\ref{4/12}) are the \emph{Euler-Lagrange equations} for critical
points of the action functional corresponding to the Lagrangian $L.$%
\smallskip

\subsection{De Donder equations}

Let $\Xi $ be the boundary form of $d\Lambda ,$ and let 
\begin{equation}
\Theta =\pi _{\mathrm{k}}^{\mathrm{2k-1}\ast }\Lambda +\Xi .  \label{4/13}
\end{equation}%
Equation (\ref{4/13}) generalizes the construction of De Donder \cite{de
donder 1929} to $\mathrm{k}>1$. We refer to $\Theta $ as a De Donder form of 
$\Lambda $. It follows from Theorem \ref{Theorem 1} that $\Theta $
satisifies the following conditions.\smallskip

\begin{corollary}
\begin{enumerate}
\item $\Theta $ is semi-basic with respect basic to the forgetful map $\pi _{%
\mathrm{k-1}}^{\mathrm{2k-1}}:J^{\mathrm{2k-1}}(M,N)\rightarrow J^{\mathrm{k}%
}(M,N).$ In other words, for any vector field $X$ tangent to fibres of $\pi
_{\mathrm{k-1}}^{\mathrm{2k-1}}:J^{\mathrm{2k-1}}(M,N)\rightarrow J^{\mathrm{%
k}}(M,N)$, 
\begin{equation}
X%
%TCIMACRO{%
%\TeXButton{lefthook}{{\mbox{$ \rule {5pt} {.5pt}\rule {.5pt} {6pt} \, $}}}}%
%BeginExpansion
{\mbox{$ \rule {5pt} {.5pt}\rule {.5pt} {6pt} \, $}}%
%EndExpansion
\Theta =0.  \label{4/14}
\end{equation}

\item For every vector field $X$ on $J^{\mathrm{2k-1}}(M,N)$ tangent to
fibres of the source map $\pi ^{\mathrm{2k-1}}:J^{\mathrm{2k-1}%
}(M,N)\rightarrow M,$ the left interior product $X%
%TCIMACRO{%
%\TeXButton{lefthook}{{\mbox{$ \rule {5pt} {.5pt}\rule {.5pt} {6pt} \, $}}}}%
%BeginExpansion
{\mbox{$ \rule {5pt} {.5pt}\rule {.5pt} {6pt} \, $}}%
%EndExpansion
\Theta $ is semi-basic with respect to the source map. In other words, for
every pair $X_{1},X_{2}$ of vector fields on $J^{\mathrm{2k-1}}(M,N)$
tangent to fibres of of the source map $\pi ^{\mathrm{2k-1}}:J^{\mathrm{2k-1}%
}(M,N)\rightarrow M,$%
\begin{equation}
X_{2}%
%TCIMACRO{%
%\TeXButton{lefthook}{{\mbox{$ \rule {5pt} {.5pt}\rule {.5pt} {6pt} \, $}}}}%
%BeginExpansion
{\mbox{$ \rule {5pt} {.5pt}\rule {.5pt} {6pt} \, $}}%
%EndExpansion
\left( X_{1}%
%TCIMACRO{%
%\TeXButton{lefthook}{{\mbox{$ \rule {5pt} {.5pt}\rule {.5pt} {6pt} \, $}}}}%
%BeginExpansion
{\mbox{$ \rule {5pt} {.5pt}\rule {.5pt} {6pt} \, $}}%
%EndExpansion
\Theta \right) =0.  \label{4/15}
\end{equation}

\item For every section $\sigma $ of $\pi :N\rightarrow M$,%
\begin{equation}
j^{\mathrm{2k-1}}\sigma ^{\ast }\Theta =j^{\mathrm{k}}\sigma ^{\ast }\Lambda
.  \label{4/16}
\end{equation}

\item For every vector field $X$ on $J^{\mathrm{2k-1}}(M,N)$ tangent to
fibres of the target map $\pi _{0}^{\mathrm{2k-1}}:J^{\mathrm{2k-1}%
}(M,N)\rightarrow N,$ and every section $\sigma $ of $\pi :N\rightarrow M,$%
\begin{equation}
j^{\mathrm{2k-1}}\sigma ^{\ast }\left( X%
%TCIMACRO{%
%\TeXButton{lefthook}{{\mbox{$ \rule {5pt} {.5pt}\rule {.5pt} {6pt} \, $}}}}%
%BeginExpansion
{\mbox{$ \rule {5pt} {.5pt}\rule {.5pt} {6pt} \, $}}%
%EndExpansion
\mathrm{d}\Theta \right) =0.  \label{4/17}
\end{equation}%
\smallskip
\end{enumerate}
\end{corollary}

\begin{theorem}
\label{Theorem 2} For $\sigma \in S^{\infty }(\bar{K},N),$ suppose that $j^{%
\mathrm{2k-1}}\sigma (\bar{K})$ is in the domain of a De Donder form $\Theta 
$. Then $\sigma $ is a critical section of the functional $A$, given by
equation (\ref{4/1}), if and only if 
\begin{equation}
j^{\mathrm{2k-1}}\sigma ^{\ast }\left( X%
%TCIMACRO{%
%\TeXButton{lefthook}{{\mbox{$ \rule {5pt} {.5pt}\rule {.5pt} {6pt} \, $}}}}%
%BeginExpansion
{\mbox{$ \rule {5pt} {.5pt}\rule {.5pt} {6pt} \, $}}%
%EndExpansion
\mathrm{d}\Theta \right) =0  \label{4/18}
\end{equation}%
for every vector field $X$ on $J^{\mathrm{2k-1}}(M,N)$ that is tangent to
fibres of the source map $\pi ^{\mathrm{2k-1}}:J^{\mathrm{2k-1}%
}(M,N)\rightarrow M.$ \smallskip
\end{theorem}

\begin{proof}
Equation (\ref{4/16}) implies that replacing $j^{\mathrm{k}}\sigma ^{\ast
}\Lambda $ by $j^{\mathrm{2k-1}}\sigma ^{\ast }\Theta $ in equation (\ref%
{4/1}) does not change the action functional,%
\begin{equation}
A(\sigma )=\int_{K}j^{k}\sigma ^{\ast }\Lambda =\int_{j^{2k-1}\sigma
(K)}\Theta .  \label{4/19}
\end{equation}%
Moreover, if $Y$ is a vertical vector field on $N$, then 
\begin{equation}
j^{\mathrm{2k-1}}\sigma ^{\ast }\pounds _{Y^{\mathrm{2k-1}}}\Theta =j^{%
\mathrm{k}}\sigma ^{\ast }\pounds _{Y^{\mathrm{k}}}\Lambda ,  \label{4/20}
\end{equation}%
where $Y^{\mathrm{2k-1}}$ is the prolongation of $Y$ to $J^{\mathrm{2k-1}%
}(M,N).$ Hence, $\sigma $ is a critical section of the functional $A$ if 
\begin{equation}
\int_{K}j^{\mathrm{2k-1}}\sigma ^{\ast }\pounds _{Y^{\mathrm{2k-1}}}\Theta =0
\label{4/21}
\end{equation}%
for every vertical vector field $Y$ on $N$ such that $Y^{\mathrm{k-1}}$
vanishes on $j^{\mathrm{k-1}}\sigma (\partial K)$. The argument leading from
equation (\ref{4/2}) to equation (\ref{4/5}) ensures that $\sigma $ is a
critical section of $A$ if and only if 
\begin{equation}
\int_{K}j^{\mathrm{2k-1}}\sigma ^{\ast }\left( Y^{\mathrm{2k-1}}%
%TCIMACRO{%
%\TeXButton{lefthook}{{\mbox{$ \rule {5pt} {.5pt}\rule {.5pt} {6pt} \, $}}}}%
%BeginExpansion
{\mbox{$ \rule {5pt} {.5pt}\rule {.5pt} {6pt} \, $}}%
%EndExpansion
\mathrm{d}\Theta \right) =0  \label{4/22}
\end{equation}%
for all vertical vector fields $Y$ on $N,$ such that $Y^{\mathrm{k-1}}$
vanishes on $j^{\mathrm{k-1}}\sigma ^{\ast }(\partial K)$.\smallskip

Equation (\ref{4/17}) ensures that in equation (\ref{4/22}), we can replace $%
Y^{\mathrm{2k-1}}$ by an arbitrary vector field $X$ on $J^{\mathrm{2k-1}%
}(M,N)$ that is tangent to fibres of the source map $\pi ^{\mathrm{2k-1}}:J^{%
\mathrm{2k-1}}(M,N)\rightarrow M$ and satisfies the condition $T\pi _{%
\mathrm{k-1}}^{\mathrm{2k-1}}\circ X\circ j^{\mathrm{k-1}}\sigma (\partial
K)=0$. In other words, we may omit the requirement that $Y^{\mathrm{2k-1}}$
is the prolongation of a vertical vector field $Y$ on $N$. This proves that
that $\sigma $ is a critical section of $A_{K}$ if and only if 
\begin{equation}
\int_{K}j^{\mathrm{2k-1}}\sigma ^{\ast }\left( X%
%TCIMACRO{%
%\TeXButton{lefthook}{{\mbox{$ \rule {5pt} {.5pt}\rule {.5pt} {6pt} \, $}}}}%
%BeginExpansion
{\mbox{$ \rule {5pt} {.5pt}\rule {.5pt} {6pt} \, $}}%
%EndExpansion
\mathrm{d}\Theta \right) =0  \label{4/23}
\end{equation}%
for every vector field $X$ on $J^{\mathrm{2k-1}}(M,N)$ that is tangent to
fibres of the source map $\pi ^{\mathrm{2k-1}}:J^{\mathrm{2k-1}%
}(M,N)\rightarrow M$ and satisfies the condition $T\pi _{\mathrm{k-1}}^{%
\mathrm{2k-1}}\circ X\circ j^{\mathrm{k=1}}\sigma (\partial K)=0$. \smallskip

Suppose that $\sigma $ is a critical section of $A$. Equation (\ref{4/23})
and the Fundamental Theorem in the Calculus of Variations, this condition
ensure that%
\begin{equation}
j^{\mathrm{2k-1}}\sigma ^{\ast }\left( X%
%TCIMACRO{%
%\TeXButton{lefthook}{{\mbox{$ \rule {5pt} {.5pt}\rule {.5pt} {6pt} \, $}}}}%
%BeginExpansion
{\mbox{$ \rule {5pt} {.5pt}\rule {.5pt} {6pt} \, $}}%
%EndExpansion
\mathrm{d}\Theta \right) =0  \label{4/24}
\end{equation}%
for every vector field $X$ on $J^{\mathrm{2k-1}}(M,N)$ that is tangent to
fibres of the source map $\pi ^{\mathrm{2k-1}}:J^{\mathrm{2k-1}%
}(M,N)\rightarrow M$.

Conversely, assume that equation (\ref{4/24}) is satisfied for all vector
fields on $J^{\mathrm{2k-1}}(M,N)$ that are tangent to fibres of the source
map $\pi ^{\mathrm{2k-1}}:J^{\mathrm{2k-1}}(M,N)\rightarrow M$. Then,
equation (\ref{4/23}) is satisfied for every vector field $X$ on $J^{\mathrm{%
2k-1}}(M,N)$ because the integrand is identically zero. In particular,
equation (\ref{4/22}) is satisfied for prolongations $Y^{\mathrm{2k-1}}$ of
vertical vector fields $Y$ on $N$ that vanish on $\partial K$ together with
all derivatives up to order $\mathrm{k}$. This ensures that $\sigma $ is a
critical point of $A$. \smallskip
\end{proof}

We refer to (\ref{4/18}) and (\ref{4/24}) as De Donder equations. They are a
system of equations in differential forms that is equivalent to
Euler-Lagrange equations. \smallskip

Note that Condition 4 in Corollary 8 on De Donder form $\Theta ,$ see
equation (\ref{4/17}), differs from equation (\ref{4/24}) only by
restriction on the range of the vector field $X$. We can combine these two
conditions in the corollary below.\smallskip

\begin{corollary}
\label{Corollary 10} A section $\sigma \in S^{\infty }(\bar{K},N)$ is a
critical section of the functional $A$, given by equation (\ref{4/1}), if
there exists a boundary form $\Xi $ such that 
\begin{equation}
j^{\mathrm{2k-1}}\sigma ^{\ast }\left( X%
%TCIMACRO{%
%\TeXButton{lefthook}{{\mbox{$ \rule {5pt} {.5pt}\rule {.5pt} {6pt} \, $}}}}%
%BeginExpansion
{\mbox{$ \rule {5pt} {.5pt}\rule {.5pt} {6pt} \, $}}%
%EndExpansion
\mathrm{d}\left( \pi _{\mathrm{k}}^{\mathrm{2k-1}\ast }\Lambda +\Xi \right)
\right) =0  \label{4/24a}
\end{equation}%
for every vector field $X$ on $J^{\mathrm{2k-1}}(M,N)$ that is tangent to
fibres of the source map $\pi ^{\mathrm{2k-1}}:J^{\mathrm{2k-1}%
}(M,N)\rightarrow M$. \smallskip
\end{corollary}

Equation (\ref{4/24a}) is a relation in the space of pairs $(\Xi ,\sigma )$.
However, it is not in the form of the symplectic relation occurring in
Tulczyjew triples. For a discussion of Tulczyjew triples in higher
derivative field theory see reference \cite{grabowska-vitagliano}.\smallskip

Since boundary forms are defined only locally, the assumption in Theorem \ref%
{Theorem 2} appears to be quite restrictive. We show that this is not the
case.\smallskip\ 

\begin{proposition}
\label{Proposition 2} A section $\sigma \in S^{\infty }(\bar{K},N)$ is a
critical section of the functional $A$ if there exists an open cover $%
\{U_{\alpha }\}$ of $\bar{K}\subset M$ such that $j^{\mathrm{2k-1}}\sigma
(U_{\alpha })$ is in the domain of a De Donder form $\Theta _{\alpha }$, and 
\begin{equation}
j^{\mathrm{2k-1}}\sigma ^{\ast }\left( X%
%TCIMACRO{%
%\TeXButton{lefthook}{{\mbox{$ \rule {5pt} {.5pt}\rule {.5pt} {6pt} \, $}}}}%
%BeginExpansion
{\mbox{$ \rule {5pt} {.5pt}\rule {.5pt} {6pt} \, $}}%
%EndExpansion
\mathrm{d}\Theta _{\alpha }\right) _{\mid K\cap U_{\alpha }}=0  \label{4/25}
\end{equation}%
for each $\alpha $ , and every vector field $X$ on $J^{\mathrm{2k-1}}(M,N).$
\smallskip
\end{proposition}

\begin{proof}
Corollary \ref{Corollary 1} ensures that, if $\Xi $ and $\Xi ^{\prime }$ are
boundary forms of $\mathrm{d}\Lambda $ with the same domain and $\Theta $
and $\Theta ^{\prime }$ are the De Donder forms corresponding to $\Xi $ and $%
\Xi ^{\prime }$, respectively, then 
\begin{equation*}
j^{\mathrm{2k-1}}\sigma ^{\ast }\left( X%
%TCIMACRO{%
%\TeXButton{lefthook}{{\mbox{$ \rule {5pt} {.5pt}\rule {.5pt} {6pt} \, $}}}}%
%BeginExpansion
{\mbox{$ \rule {5pt} {.5pt}\rule {.5pt} {6pt} \, $}}%
%EndExpansion
\mathrm{d}\Theta ^{\prime }\right) =j^{\mathrm{2k-1}}\sigma ^{\ast }\left( X%
%TCIMACRO{%
%\TeXButton{lefthook}{{\mbox{$ \rule {5pt} {.5pt}\rule {.5pt} {6pt} \, $}}}}%
%BeginExpansion
{\mbox{$ \rule {5pt} {.5pt}\rule {.5pt} {6pt} \, $}}%
%EndExpansion
\mathrm{d}\Theta \right)
\end{equation*}%
for each section $\sigma $ of $\pi $ and every vector field $X$ on $J^{%
\mathrm{2k-1}}(M,N)$. Hence, the choice of a De Donder form does not matter.

For each $\alpha $, equation (\ref{4/25}) is equivalent to Euler-Lagrange
equations for $\sigma $ in $K\cap U_{\alpha }$. Since Euler-Lagrange
equations are local, the conditions of Proposition \ref{Proposition 2} imply
that $\sigma $ satisfies Euler-Lagrange equations in $K$. \smallskip
\end{proof}

\subsection{Symmetries and conservation laws}

\begin{definition}
A vector field $Y$ on $N$ is an infinitesimal symmetry of the Lagrangian
system with Lagrangian $\Lambda =L\mathrm{d}_{m}x$ of differential order $%
\mathrm{k}$ if it projects to a vector field on $M$ and $\pounds _{Y^{%
\mathrm{k}}}\mathrm{d}\Lambda =0$.
\end{definition}

Let $Y$ be an infinitesimal symmetry of $\Lambda $. For every boundary form $%
\Xi $ of $\mathrm{d}\Lambda $, Lemma \ref{Lemma 1} ensures that $j^{\mathrm{%
2k-1}}\sigma ^{\ast }(\pounds _{Y^{\mathrm{2k-1}}}\Xi )=0$ for all sections $%
\sigma $ of $\pi :M\rightarrow N$. Since $\Theta =\pi _{\mathrm{k}}^{%
\mathrm{2k-1}\ast }\Lambda +\Xi $ is the local De Donder form corresponding
to $\Xi $, it follows that 
\begin{equation}
j^{\mathrm{2k-1}}\sigma ^{\ast }(\pounds _{Y^{\mathrm{2k-1}}}\Theta )=0.
\label{5.26}
\end{equation}%
Hence, 
\begin{equation}
j^{\mathrm{2k-1}}\sigma ^{\ast }(Y^{\mathrm{2k-1}}%
%TCIMACRO{%
%\TeXButton{lefthook}{{\mbox{$ \rule {5pt} {.5pt}\rule {.5pt} {6pt} \, $}}}}%
%BeginExpansion
{\mbox{$ \rule {5pt} {.5pt}\rule {.5pt} {6pt} \, $}}%
%EndExpansion
\mathrm{d}\Theta )+j^{\mathrm{2k-1}}\sigma ^{\ast }[\mathrm{d}(Y^{\mathrm{%
2k-1}}%
%TCIMACRO{%
%\TeXButton{lefthook}{{\mbox{$ \rule {5pt} {.5pt}\rule {.5pt} {6pt} \, $}}}}%
%BeginExpansion
{\mbox{$ \rule {5pt} {.5pt}\rule {.5pt} {6pt} \, $}}%
%EndExpansion
\Theta )]=0.  \label{5.27}
\end{equation}%
If $\sigma $ satisfies De Donder equations, we get the conservation law 
\begin{equation}
\mathrm{d}\left[ j^{\mathrm{2k-1}}\sigma ^{\ast }\left( Y^{\mathrm{2k-1}}%
%TCIMACRO{%
%\TeXButton{lefthook}{{\mbox{$ \rule {5pt} {.5pt}\rule {.5pt} {6pt} \, $}}}}%
%BeginExpansion
{\mbox{$ \rule {5pt} {.5pt}\rule {.5pt} {6pt} \, $}}%
%EndExpansion
\Theta \right) \right] =0.  \label{5.28}
\end{equation}

If $K$ is an open, relatively compact submanifold of $M$ with boundary $%
\partial K$, which is contained in $\mathrm{domain~}\sigma ,$ then%
\begin{equation}
\int_{\partial K}j^{\mathrm{2k-1}}\sigma ^{\ast }\left( Y^{\mathrm{2k-1}}%
%TCIMACRO{%
%\TeXButton{lefthook}{{\mbox{$ \rule {5pt} {.5pt}\rule {.5pt} {6pt} \, $}}}}%
%BeginExpansion
{\mbox{$ \rule {5pt} {.5pt}\rule {.5pt} {6pt} \, $}}%
%EndExpansion
\Theta \right) =0.  \label{5.29}
\end{equation}%
In other words, if $\partial K=\Sigma _{1}\cup \Sigma _{2}$, where $\Sigma
_{1}$ and $\Sigma _{2}$ inherit outer orientation from $\partial K,$ and $%
\Sigma _{1}\cap \Sigma _{2}$ is smooth of dimension $n-2$, then 
\begin{equation}
\int_{\Sigma _{1}}j^{\mathrm{2k-1}}\sigma ^{\ast }\left( Y^{\mathrm{2k-1}}%
%TCIMACRO{%
%\TeXButton{lefthook}{{\mbox{$ \rule {5pt} {.5pt}\rule {.5pt} {6pt} \, $}}}}%
%BeginExpansion
{\mbox{$ \rule {5pt} {.5pt}\rule {.5pt} {6pt} \, $}}%
%EndExpansion
\Theta \right) =\int_{\Sigma _{2}}j^{\mathrm{2k-1}}\sigma ^{\ast }\left( Y^{%
\mathrm{2k-1}}%
%TCIMACRO{%
%\TeXButton{lefthook}{{\mbox{$ \rule {5pt} {.5pt}\rule {.5pt} {6pt} \, $}}}}%
%BeginExpansion
{\mbox{$ \rule {5pt} {.5pt}\rule {.5pt} {6pt} \, $}}%
%EndExpansion
\Theta \right) .  \label{5.30}
\end{equation}%
If the De Donder equations are hyperbolic, and $\Sigma _{1}\ $and $\Sigma
_{2}$ are Cauchy surfaces, then the integrals in equation (\ref{5.30}) are
conserved quantities corresponding to the infinitesimal symmetry $Y$.
\smallskip

A priori, the integrals on each side of equation (\ref{5.30}) depend on the
choice of the boundary form $\Xi $. However, the difference between the left
and the right hand sides of equation (\ref{5.30}) vanishes for every $\Xi $.
In an example below, we show how boundary conditions lead to unique
expressions for constants of motion.

\section{Example}

\subsection{Cauchy problem}

Consider $M=\mathbb{R}^{2}$ with coordinates $\boldsymbol{x}=(x^{1},x^{2})$
and $N=T\mathbb{R}^{2}$ with coordinates $(\boldsymbol{x},\boldsymbol{y}%
)=(x^{1},x^{2},y^{1},y^{2})$. 
\begin{equation}
L=g_{ab}g^{ij}g^{kl}z_{ij}^{a}z_{kl}^{b},  \label{7.1}
\end{equation}%
where $g_{ab}$ is the Minkowski metric. 
\begin{eqnarray}
\frac{\partial L}{\partial z_{ij}^{a}} &=&2g_{ab}g^{ij}g^{kl}z_{kl}^{b},
\label{7.2} \\
\frac{\partial L}{\partial z_{i}^{a}} &=&0\text{ \ and \ }\frac{\partial L}{%
\partial y^{a}}=0.  \notag
\end{eqnarray}

Euler-Lagrange equations 
\begin{eqnarray*}
\frac{\partial ^{2}}{\partial x^{i}\partial x^{j}}\frac{\partial L}{\partial
y_{,ij}^{a}} &=&0 \\
2g_{ab}g^{ij}g^{kl}y_{,klij}^{a} &=&0\text{. }
\end{eqnarray*}%
Writing $x^{1}=t$, $x^{2}=x$, we get 
\begin{eqnarray*}
\left( \frac{\partial ^{2}}{\partial t^{2}}-\frac{\partial ^{2}}{\partial x^{2}}%
\right) \left( \frac{\partial ^{2}}{\partial t^{2}}-\frac{\partial ^{2}}{%
\partial x^{2}}\right) y^{a}(t,x) &=&0. \\
\left( \frac{\partial ^{4}}{\partial t^{4}}-2\frac{\partial ^{2}}{\partial
t^{2}}\frac{\partial ^{2}}{\partial x^{2}}+\frac{\partial ^{4}}{\partial
x^{2}}\right) y^{a}(t,x) &=&0
\end{eqnarray*}%
Set $y(t,x)$, $\dot{y}(t,x)$, $\ddot{y}(t,x),$ and $\dddot{y}(t,x)$ as the
Cauchy data at $t.$ Then%
\begin{eqnarray*}
\frac{\partial }{\partial t}y^{a}(t,x) &=&\dot{y}^{a}(t,x), \\
\frac{\partial }{\partial t}\dot{y}^{a}(t,x) &=&\ddot{y}^{a}(t,x), \\
\frac{\partial }{\partial t}\ddot{y}^{a}(t,x) &=&\dddot{y}^{a}(t,x), \\
\frac{\partial }{\partial t}\dddot{y}^{a}(t,x) &=&\frac{\partial ^{4}}{%
\partial t^{4}}y^{a}(t,x)=2\frac{\partial ^{2}}{\partial x^{2}}\ddot{y}%
^{a}(t,x)-\frac{\partial ^{4}}{\partial x^{4}}y^{a}(t,x).
\end{eqnarray*}%
Therefore, 
\begin{equation*}
\frac{\partial }{\partial t}\left( 
\begin{array}{c}
y^{a} \\ 
\dot{y}^{a} \\ 
\ddot{y}^{a} \\ 
\dddot{y}^{a}%
\end{array}%
\right) =\left( 
\begin{array}{c}
\dot{y}^{a} \\ 
\ddot{y}^{a} \\ 
\dddot{y}^{a} \\ 
2\frac{\partial ^{2}}{\partial x^{2}}\ddot{y}^{a}-\frac{\partial ^{4}}{%
\partial x^{4}}y^{a}%
\end{array}%
\right) =A\left( 
\begin{array}{c}
y^{a} \\ 
\dot{y}^{a} \\ 
\ddot{y}^{a} \\ 
\dddot{y}^{a}%
\end{array}%
\right)
\end{equation*}

where 
\begin{equation*}
A=\left( 
\begin{array}{cccc}
0 & 1 & 0 & 0 \\ 
0 & 0 & 1 & 0 \\ 
0 & 0 & 0 & 1 \\ 
-\frac{\partial ^{4}}{\partial x^{4}} & 0 & 2\frac{\partial ^{2}}{\partial
x^{2}} & 0%
\end{array}%
\right) .
\end{equation*}%
Since 
\begin{equation*}
e^{tA}=\sum_{k=0}^{\infty }\frac{t^{k}A^{k}}{k!}
\end{equation*}%
is well defined and 
\begin{equation*}
\left( 
\begin{array}{c}
y^{a}(t,x) \\ 
\dot{y}^{a}(t,x) \\ 
\ddot{y}^{a}(t,x) \\ 
\dddot{y}^{a}(t,x)%
\end{array}%
\right) =\sum_{k=0}^{\infty }\frac{t^{k}A^{k}}{k!}\left( 
\begin{array}{c}
y^{a}(0,x) \\ 
\dot{y}^{a}(0,x) \\ 
\ddot{y}^{a}(0,x) \\ 
\dddot{y}^{a}(0,x)%
\end{array}%
\right)
\end{equation*}%
is a solution of the Cauchy problem at $t=0.$

\subsection{De Donder forms}

De Donder forms are $\pi _{\mathrm{2}}^{\mathrm{3}\ast }\Lambda +\Xi $,
where 
\begin{equation}
\Xi =p_{a}^{i}(\mathrm{d}y^{a}-z_{j}^{a}\mathrm{d}x^{j})\wedge \left( \frac{%
{\small \partial }}{{\small \partial x}^{i}}%
%TCIMACRO{%
%\TeXButton{lefthook}{{\mbox{$ \rule {5pt} {.5pt}\rule {.5pt} {6pt} \, $}}}}%
%BeginExpansion
{\mbox{$ \rule {5pt} {.5pt}\rule {.5pt} {6pt} \, $}}%
%EndExpansion
\mathrm{d}_{2}x\right) +p_{a}^{i_{1}i_{2}}(\mathrm{d}z_{i%
\,_{2}}^{a}-z_{i_{2}j}^{a}\mathrm{d}x^{j})\wedge \left( \frac{{\small %
\partial }}{{\small \partial x}^{i_{1}}}%
%TCIMACRO{%
%\TeXButton{lefthook}{{\mbox{$ \rule {5pt} {.5pt}\rule {.5pt} {6pt} \, $}}}}%
%BeginExpansion
{\mbox{$ \rule {5pt} {.5pt}\rule {.5pt} {6pt} \, $}}%
%EndExpansion
\mathrm{d}_{2}x\right)  \label{7.3}
\end{equation}%
is a local boundary form corresponding to $\mathrm{d}\Lambda .$ For a
section $\sigma $ of $\pi $,%
\begin{eqnarray}
&&\Theta _{\mid \mathrm{range}~j^{\mathrm{3}}\sigma }=\pi _{\mathrm{2}}^{%
\mathrm{3}\ast }\Lambda _{\mid \mathrm{range}~j^{\mathrm{3}}\sigma }+\Xi
_{\mid \mathrm{range}~j^{\mathrm{3}}\sigma }  \label{7.4} \\
&=&L\mathrm{d}_{2}x+P_{a}^{i_{1}}(\mathrm{d}y^{a}-z_{j}^{a}\mathrm{d}%
x^{j})\wedge \left( \frac{{\small \partial }}{{\small \partial x}^{i_{1}}}%
%TCIMACRO{%
%\TeXButton{lefthook}{{\mbox{$ \rule {5pt} {.5pt}\rule {.5pt} {6pt} \, $}}}}%
%BeginExpansion
{\mbox{$ \rule {5pt} {.5pt}\rule {.5pt} {6pt} \, $}}%
%EndExpansion
\mathrm{d}_{2}x\right) +P_{a}^{i_{1}i_{2}}(\mathrm{d}z_{i%
\,_{2}}^{a}-z_{i_{2}j}^{a}\mathrm{d}x^{j})\wedge \left( \frac{{\small %
\partial }}{{\small \partial x}^{i_{1}}}%
%TCIMACRO{%
%\TeXButton{lefthook}{{\mbox{$ \rule {5pt} {.5pt}\rule {.5pt} {6pt} \, $}}}}%
%BeginExpansion
{\mbox{$ \rule {5pt} {.5pt}\rule {.5pt} {6pt} \, $}}%
%EndExpansion
\mathrm{d}_{2}x\right)  \notag \\
&=&P_{a}^{i_{1}}\mathrm{d}y^{a}\wedge \left( \frac{{\small \partial }}{%
{\small \partial x}^{i_{1}}}%
%TCIMACRO{%
%\TeXButton{lefthook}{{\mbox{$ \rule {5pt} {.5pt}\rule {.5pt} {6pt} \, $}}}}%
%BeginExpansion
{\mbox{$ \rule {5pt} {.5pt}\rule {.5pt} {6pt} \, $}}%
%EndExpansion
\mathrm{d}_{2}x\right) +P_{a}^{i_{1}i_{2}}\mathrm{d}z_{i\,_{2}}^{a}\wedge
\left( \frac{{\small \partial }}{{\small \partial x}^{i_{1}}}%
%TCIMACRO{%
%\TeXButton{lefthook}{{\mbox{$ \rule {5pt} {.5pt}\rule {.5pt} {6pt} \, $}}}}%
%BeginExpansion
{\mbox{$ \rule {5pt} {.5pt}\rule {.5pt} {6pt} \, $}}%
%EndExpansion
\mathrm{d}_{2}x\right) -\left(
P_{a}^{i}y_{,i}^{a}+P_{a}^{ij}y_{,ij}-L\right) \mathrm{d}_{2}x.  \notag
\end{eqnarray}%
where the functions $P_{a}^{i_{1}}$ and $P_{a}^{i_{1}i_{2}}$ satisfy the
equations 
\begin{eqnarray}
\frac{\partial L}{\partial z_{ij}^{a}}-P_{a}^{(i_{1}i_{2})} &=&0,
\label{7.5} \\
\frac{\partial L}{\partial z_{i}^{a}}-P_{a}^{i}-P_{a,i_{1}}^{i_{1}i} &=&0, 
\notag
\end{eqnarray}%
that follow from equations (\ref{14}).

Since 
\begin{equation*}
\mathrm{d}\Lambda =\frac{\partial L}{\partial z_{ij}^{a}}\mathrm{d}%
z_{ij}^{a}\wedge \mathrm{d}_{2}x+\frac{\partial L}{\partial z_{i}^{a}}%
\mathrm{d}z_{1}^{a}\wedge \mathrm{d}_{2}x+\frac{\partial L}{\partial y^{a}}%
\mathrm{d}y^{a}\wedge \mathrm{d}_{2}x=2g_{ab}g^{ij}g^{kl}z_{kl}^{b}~\mathrm{d%
}z_{ij}^{a}\wedge \mathrm{d}_{2}x,
\end{equation*}%
it follows that 
\begin{eqnarray*}
\frac{\partial L}{\partial z_{ij}^{a}}
&=&2g_{ab}g^{i_{1}i_{2}}g^{kl}y_{,kl}^{b}, \\
\frac{\partial L}{\partial z_{i}^{a}} &=&0.
\end{eqnarray*}%
Hence, the symmetric solution is 
\begin{eqnarray*}
P_{a}^{i_{1}i_{2}} &=&\Phi
_{a}^{i_{1}i_{2}}=2g_{ab}g^{i_{1}i_{2}}g^{kl}y_{,kl}^{b}, \\
P_{a,i_{1}}^{i_{1}i_{2}} &=&\Phi
_{a,i_{1}}^{i_{1}i_{2}}=2g_{ab}g^{i_{1}i_{2}}g^{kl}y_{,kli_{1}}^{b}, \\
P_{a}^{i} &=&\Phi
_{a}^{i}-P_{a,i_{1}}^{i_{1}i}=-2g_{ab}g^{i_{1}i_{2}}g^{kl}y_{,kli_{1}}^{b}.
\end{eqnarray*}%
In this case 
\begin{eqnarray*}
&&\Xi _{\mid \mathrm{range}~j^{\mathrm{3}}\sigma }=P_{a}^{i_{1}i_{2}}(%
\mathrm{d}y^{a}-z_{j}^{a}\mathrm{d}x^{j})\wedge \left( \frac{{\small %
\partial }}{{\small \partial x}^{i}}%
%TCIMACRO{%
%\TeXButton{lefthook}{{\mbox{$ \rule {5pt} {.5pt}\rule {.5pt} {6pt} \, $}}}}%
%BeginExpansion
{\mbox{$ \rule {5pt} {.5pt}\rule {.5pt} {6pt} \, $}}%
%EndExpansion
\mathrm{d}_{2}x\right)  \\
&&+P_{a}^{i_{1}i_{2}}(\mathrm{d}z_{i\,_{2}}^{a}-z_{i_{2}j}^{a}\mathrm{d}%
x^{j})\wedge \left( \frac{{\small \partial }}{{\small \partial x}^{i_{1}}}%
%TCIMACRO{%
%\TeXButton{lefthook}{{\mbox{$ \rule {5pt} {.5pt}\rule {.5pt} {6pt} \, $}}}}%
%BeginExpansion
{\mbox{$ \rule {5pt} {.5pt}\rule {.5pt} {6pt} \, $}}%
%EndExpansion
\mathrm{d}_{2}x\right)  \\
&=&-2g_{ab}g^{i_{1}i_{2}}g^{kl}y_{,kli_{1}}^{b}(\mathrm{d}y^{a}-z_{j}^{a}%
\mathrm{d}x^{j})\wedge \left( \frac{{\small \partial }}{{\small \partial x}%
^{i_{2}}}%
%TCIMACRO{%
%\TeXButton{lefthook}{{\mbox{$ \rule {5pt} {.5pt}\rule {.5pt} {6pt} \, $}}}}%
%BeginExpansion
{\mbox{$ \rule {5pt} {.5pt}\rule {.5pt} {6pt} \, $}}%
%EndExpansion
\mathrm{d}_{2}x\right)  \\
&&+2g_{ab}g^{i_{1}i_{2}}g^{kl}y_{,kl}^{b}(\mathrm{d}z_{i%
\,_{2}}^{a}-z_{i_{2}j}^{a}\mathrm{d}x^{j})\wedge \left( \frac{{\small %
\partial }}{{\small \partial x}^{i_{1}}}%
%TCIMACRO{%
%\TeXButton{lefthook}{{\mbox{$ \rule {5pt} {.5pt}\rule {.5pt} {6pt} \, $}}}}%
%BeginExpansion
{\mbox{$ \rule {5pt} {.5pt}\rule {.5pt} {6pt} \, $}}%
%EndExpansion
\mathrm{d}_{2}x\right) ,
\end{eqnarray*}%
and 
\begin{eqnarray}
&&\Theta _{\mid \mathrm{range}~j^{\mathrm{3}}\sigma }=\pi _{2}^{3\ast
}\Lambda _{\mid \mathrm{range}~j^{\mathrm{3}}\sigma }+\Xi _{\mid \mathrm{%
range}~j^{\mathrm{3}}\sigma }  \label{7.5a} \\
&=&Ld_{2}x+P_{a}^{i_{1}i_{2}}(\mathrm{d}y^{a}-z_{j}^{a}\mathrm{d}%
x^{j})\wedge \left( \frac{{\small \partial }}{{\small \partial x}^{i}}%
%TCIMACRO{%
%\TeXButton{lefthook}{{\mbox{$ \rule {5pt} {.5pt}\rule {.5pt} {6pt} \, $}}}}%
%BeginExpansion
{\mbox{$ \rule {5pt} {.5pt}\rule {.5pt} {6pt} \, $}}%
%EndExpansion
\mathrm{d}_{2}x\right)   \notag \\
&&+P_{a}^{i_{1}i_{2}}(\mathrm{d}z_{i\,_{2}}^{a}-z_{i_{2}j}^{a}\mathrm{d}%
x^{j})\wedge \left( \frac{{\small \partial }}{{\small \partial x}^{i_{1}}}%
%TCIMACRO{%
%\TeXButton{lefthook}{{\mbox{$ \rule {5pt} {.5pt}\rule {.5pt} {6pt} \, $}}}}%
%BeginExpansion
{\mbox{$ \rule {5pt} {.5pt}\rule {.5pt} {6pt} \, $}}%
%EndExpansion
\mathrm{d}_{2}x\right)   \notag \\
&=&g_{ab}g^{ij}g^{kl}z_{ij}^{a}z_{kl}^{b}\mathrm{d}%
_{2}x-2g_{ab}g^{i_{1}i_{2}}g^{kl}y_{,kli_{1}}^{b}(\mathrm{d}y^{a}-z_{j}^{a}%
\mathrm{d}x^{j})\wedge \left( \frac{{\small \partial }}{{\small \partial x}%
^{i_{2}}}%
%TCIMACRO{%
%\TeXButton{lefthook}{{\mbox{$ \rule {5pt} {.5pt}\rule {.5pt} {6pt} \, $}}}}%
%BeginExpansion
{\mbox{$ \rule {5pt} {.5pt}\rule {.5pt} {6pt} \, $}}%
%EndExpansion
\mathrm{d}_{2}x\right)   \notag \\
&&+2g_{ab}g^{i_{1}i_{2}}g^{kl}y_{,kl}^{b}(\mathrm{d}z_{i%
\,_{2}}^{a}-z_{i_{2}j}^{a}\mathrm{d}x^{j})\wedge \left( \frac{{\small %
\partial }}{{\small \partial x}^{i_{1}}}%
%TCIMACRO{%
%\TeXButton{lefthook}{{\mbox{$ \rule {5pt} {.5pt}\rule {.5pt} {6pt} \, $}}}}%
%BeginExpansion
{\mbox{$ \rule {5pt} {.5pt}\rule {.5pt} {6pt} \, $}}%
%EndExpansion
\mathrm{d}_{2}x\right) .  \notag
\end{eqnarray}

A non-symmetric solution of equation (\ref{7.5}) is 
\begin{eqnarray}
P_{a}^{\prime i_{1}i_{2}}
&=&P_{a}^{i_{1}i_{2}}+Q_{a}^{i_{1}i_{2}}=2g_{ab}g^{i_{1}i_{2}}g^{kl}y_{,kl}^{b}+Q_{a}^{i_{1}i_{2}},
\label{7.5b} \\
P_{a,i_{1}}^{\prime i_{1}i_{2}}
&=&P_{a,i_{1}}^{i_{1}i_{2}}+Q_{a,i_{1}}^{i_{1}i_{2}}=2g_{ab}g^{i_{1}i_{2}}g^{kl}y_{,kli_{1}}^{b}+Q_{a,i_{1}}^{i_{1}i_{2}},
\notag \\
P_{a}^{\prime i} &=&\Phi _{a}^{i}-P_{a,i_{1}}^{\prime
i_{1}i}=-2g_{ab}g^{i_{1}i_{2}}g^{kl}y_{,kli_{1}}^{b}-Q_{a,i_{1}}^{i_{1}i_{2}},
\notag
\end{eqnarray}%
where $Q_{a}^{i_{1}i_{2}}$ is skew symmetric in $i_{1}$ and $i_{2},$%
\begin{equation}
Q_{a,i_{1}}^{i_{1}i_{2}}=-Q_{a,i_{1}}^{i_{2}i_{1}}.  \label{7.11}
\end{equation}

\subsection{Symmetries}

A vector field $Y=Y^{i}\frac{\partial }{\partial x^{i}}+Y^{a}\frac{\partial 
}{\partial y^{a}}$ is a symmetry if 
\begin{equation}
\pounds _{Y^{2}}\left( L\mathrm{d}_{2}x\right) =0,  \label{7.12}
\end{equation}%
where $Y^{2}=Y^{i}\frac{\partial }{\partial x^{i}}+Y^{a}\frac{\partial }{%
\partial y^{a}}+Y_{i}^{a}\frac{\partial }{\partial z_{i}^{a}}+Y_{ij}^{a}%
\frac{\partial }{\partial z_{ij}^{a}}$ of $Y$ is the the prolongation of $Y$
to $J^{2}(M,N)$ and 
\begin{eqnarray}
Y_{i}^{a} &=&Y_{,b}^{a}z_{i}^{b}-z_{i}^{a}Y_{,j}^{i}+Y_{,i}^{a},
\label{7.13} \\
Y_{ij}^{a} &=&z_{(j}^{b}Y_{i),b}^{a}-z_{k(i}^{a}Y_{,i)}^{k}+Y_{(i,j)}^{a}, 
\notag
\end{eqnarray}%
see equations (\ref{A8}) and (\ref{A9}) in the Appendix. The Lorentz metric $%
g_{ij}\mathrm{d}x^{i}\mathrm{d}x^{j}=(\mathrm{d}t)^{2}-(\mathrm{d}x)^{2}$
occuring in the Lagrangian has Killing vector $Y_{T}=\frac{\partial }{%
\partial x^{1}},$ $Y_{S}=\frac{\partial }{\partial x^{2}}$ corresponding to
the time and the space translations, and the infinitesimal Lorentz
transformation $Y_{L}=x^{2}\frac{\partial }{\partial x^{1}}+x^{1}\frac{%
\partial }{\partial x^{2}}.$

We discuss here conservation of energy corrsponding to the time translations 
$Y_{T}$. Equations (\ref{7.13}) show that the jet components of $Y_{T}^{%
\mathrm{2}}$ and $Y_{T}^{\mathrm{3}}$ vanish, so that $Y_{T}^{\mathrm{2}}=%
\frac{\partial }{\partial t}$ and $Y_{T}^{\mathrm{3}}=\frac{\partial }{%
\partial t}$. Hence, $\ $ 
\begin{eqnarray*}
\pounds _{Y_{T}^{\mathrm{2}}}(L\mathrm{d}_{2}x) &=&\pounds _{Y_{T}^{\mathrm{2%
}}}\left( g_{ab}g^{ij}g^{kl}z_{ij}^{a}z_{kl}^{b}\mathrm{d}_{2}x\right) \\
&=&Y_{T}^{\mathrm{2}}%
%TCIMACRO{%
%\TeXButton{lefthook}{{\mbox{$ \rule {5pt} {.5pt}\rule {.5pt} {6pt} \, $}}}}%
%BeginExpansion
{\mbox{$ \rule {5pt} {.5pt}\rule {.5pt} {6pt} \, $}}%
%EndExpansion
(2g_{ab}g^{ij}g^{kl}z_{kl}^{b})\mathrm{d}z_{ij}^{a}\wedge \mathrm{d}_{2}x+%
\mathrm{d}\left[ \left( g_{ab}g^{ij}g^{kl}z_{ij}^{a}z_{kl}^{b}\right)
Y_{T}^{i_{1}}\frac{\partial }{\partial x^{i_{i}}}%
%TCIMACRO{%
%\TeXButton{lefthook}{{\mbox{$ \rule {5pt} {.5pt}\rule {.5pt} {6pt} \, $}}}}%
%BeginExpansion
{\mbox{$ \rule {5pt} {.5pt}\rule {.5pt} {6pt} \, $}}%
%EndExpansion
\mathrm{d}_{2}x\right] \\
&=&(2g_{ab}g^{ij}g^{kl}z_{kl}^{b})Y_{Tij}^{a}\mathrm{d}%
_{2}x+(2g_{ab}g^{ij}g^{kl}z_{kl}^{b})Y_{T}^{i_{1}}\mathrm{d}z_{ij}^{a}\wedge
\left( \frac{\partial }{\partial x^{i_{i}}}%
%TCIMACRO{%
%\TeXButton{lefthook}{{\mbox{$ \rule {5pt} {.5pt}\rule {.5pt} {6pt} \, $}}}}%
%BeginExpansion
{\mbox{$ \rule {5pt} {.5pt}\rule {.5pt} {6pt} \, $}}%
%EndExpansion
\mathrm{d}_{2}x\right) + \\
&&+\left[ \left( g_{ab}g^{ij}g^{kl}z_{ij}^{a}z_{kl}^{b}\right) \right] 
\mathrm{d}Y_{T}^{i_{1}}\wedge \left( \frac{\partial }{\partial x^{i_{i}}}%
%TCIMACRO{%
%\TeXButton{lefthook}{{\mbox{$ \rule {5pt} {.5pt}\rule {.5pt} {6pt} \, $}}}}%
%BeginExpansion
{\mbox{$ \rule {5pt} {.5pt}\rule {.5pt} {6pt} \, $}}%
%EndExpansion
\mathrm{d}_{2}x\right) \\
&=&0
\end{eqnarray*}

Since 
\begin{equation*}
\left( Y_{T}^{\mathrm{3}}%
%TCIMACRO{%
%\TeXButton{lefthook}{{\mbox{$ \rule {5pt} {.5pt}\rule {.5pt} {6pt} \, $}}}}%
%BeginExpansion
{\mbox{$ \rule {5pt} {.5pt}\rule {.5pt} {6pt} \, $}}%
%EndExpansion
\mathrm{d}x^{j}\right) =\delta _{1}^{j},
\end{equation*}%
\begin{equation*}
\frac{{\small \partial }}{{\small \partial x}^{i_{1}}}%
%TCIMACRO{%
%\TeXButton{lefthook}{{\mbox{$ \rule {5pt} {.5pt}\rule {.5pt} {6pt} \, $}}}}%
%BeginExpansion
{\mbox{$ \rule {5pt} {.5pt}\rule {.5pt} {6pt} \, $}}%
%EndExpansion
\mathrm{d}_{2}x=\frac{{\small \partial }}{{\small \partial x}^{i_{1}}}%
%TCIMACRO{%
%\TeXButton{lefthook}{{\mbox{$ \rule {5pt} {.5pt}\rule {.5pt} {6pt} \, $}}}}%
%BeginExpansion
{\mbox{$ \rule {5pt} {.5pt}\rule {.5pt} {6pt} \, $}}%
%EndExpansion
\mathrm{d}x^{1}\wedge \mathrm{d}x^{2}=\delta _{i_{1}}^{1}\mathrm{d}%
x^{2}-\delta _{i_{1}}^{2}\mathrm{d}x^{1},
\end{equation*}%
and 
\begin{equation*}
\left( Y_{T}^{\mathrm{3}}%
%TCIMACRO{%
%\TeXButton{lefthook}{{\mbox{$ \rule {5pt} {.5pt}\rule {.5pt} {6pt} \, $}}}}%
%BeginExpansion
{\mbox{$ \rule {5pt} {.5pt}\rule {.5pt} {6pt} \, $}}%
%EndExpansion
\left( \frac{{\small \partial }}{{\small \partial x}^{i_{1}}}%
%TCIMACRO{%
%\TeXButton{lefthook}{{\mbox{$ \rule {5pt} {.5pt}\rule {.5pt} {6pt} \, $}}}}%
%BeginExpansion
{\mbox{$ \rule {5pt} {.5pt}\rule {.5pt} {6pt} \, $}}%
%EndExpansion
\mathrm{d}_{2}x\right) \right) =\left( Y_{T}^{2}%
%TCIMACRO{%
%\TeXButton{lefthook}{{\mbox{$ \rule {5pt} {.5pt}\rule {.5pt} {6pt} \, $}}}}%
%BeginExpansion
{\mbox{$ \rule {5pt} {.5pt}\rule {.5pt} {6pt} \, $}}%
%EndExpansion
\left( \delta _{i_{1}}^{1}\mathrm{d}x^{2}-\delta _{i_{1}}^{2}\mathrm{d}%
x^{1}\right) \right) =-\delta _{i_{1}}^{2}.
\end{equation*}%
equation (\ref{7.4}) yields%
\begin{eqnarray*}
j^{\mathrm{3}}\sigma ^{\ast }\left( Y_{T}^{\mathrm{3}}%
%TCIMACRO{%
%\TeXButton{lefthook}{{\mbox{$ \rule {5pt} {.5pt}\rule {.5pt} {6pt} \, $}}}}%
%BeginExpansion
{\mbox{$ \rule {5pt} {.5pt}\rule {.5pt} {6pt} \, $}}%
%EndExpansion
\Theta \right) &=&j^{\mathrm{3}}\sigma ^{\ast }\left[ -P_{a}^{i_{1}}dy^{a}%
\wedge \left( Y_{T}^{2}%
%TCIMACRO{%
%\TeXButton{lefthook}{{\mbox{$ \rule {5pt} {.5pt}\rule {.5pt} {6pt} \, $}}}}%
%BeginExpansion
{\mbox{$ \rule {5pt} {.5pt}\rule {.5pt} {6pt} \, $}}%
%EndExpansion
\left( \frac{{\small \partial }}{{\small \partial x}^{i_{1}}}%
%TCIMACRO{%
%\TeXButton{lefthook}{{\mbox{$ \rule {5pt} {.5pt}\rule {.5pt} {6pt} \, $}}}}%
%BeginExpansion
{\mbox{$ \rule {5pt} {.5pt}\rule {.5pt} {6pt} \, $}}%
%EndExpansion
d_{2}x\right) \right) \right] \\
&&+j^{\mathrm{3}}\sigma ^{\ast }\left[ P_{a}^{i_{1}i_{2}}dz_{i\,_{2}}^{a}%
\wedge \left( Y_{T}^{2}%
%TCIMACRO{%
%\TeXButton{lefthook}{{\mbox{$ \rule {5pt} {.5pt}\rule {.5pt} {6pt} \, $}}}}%
%BeginExpansion
{\mbox{$ \rule {5pt} {.5pt}\rule {.5pt} {6pt} \, $}}%
%EndExpansion
\left( \frac{{\small \partial }}{{\small \partial x}^{i_{1}}}%
%TCIMACRO{%
%\TeXButton{lefthook}{{\mbox{$ \rule {5pt} {.5pt}\rule {.5pt} {6pt} \, $}}}}%
%BeginExpansion
{\mbox{$ \rule {5pt} {.5pt}\rule {.5pt} {6pt} \, $}}%
%EndExpansion
d_{2}x\right) \right) \right] \\
&&-\left( P_{a}^{i}y_{,i}^{a}+P_{a}^{ij}y_{,ij}-L\right) \left( Y_{T}^{2}%
%TCIMACRO{%
%\TeXButton{lefthook}{{\mbox{$ \rule {5pt} {.5pt}\rule {.5pt} {6pt} \, $}}}}%
%BeginExpansion
{\mbox{$ \rule {5pt} {.5pt}\rule {.5pt} {6pt} \, $}}%
%EndExpansion
d_{2}x\right) \\
&=&P_{a}^{2}y_{,j}^{a}dx^{j}+P_{a}^{2i_{2}}y_{,i_{2}j}^{a}dx^{j}-\left(
P_{a}^{i}y_{,i}^{a}+P_{a}^{ij}y_{,ij}-L\right) dx^{2}.
\end{eqnarray*}%
For an open relatively compact manifold $K$ with boundary $\partial K=\Sigma
-\Sigma ^{\prime },$ 
\begin{equation*}
\int_{\Sigma }j^{\mathrm{3}}\sigma ^{\ast }\left( Y_{T}^{2}%
%TCIMACRO{%
%\TeXButton{lefthook}{{\mbox{$ \rule {5pt} {.5pt}\rule {.5pt} {6pt} \, $}}}}%
%BeginExpansion
{\mbox{$ \rule {5pt} {.5pt}\rule {.5pt} {6pt} \, $}}%
%EndExpansion
\Theta \right) =\int_{\Sigma }\left[ P_{a}^{2}y_{,j}^{a}\mathrm{d}%
x^{j}+P_{a}^{2i_{2}}y_{,i_{2}j}^{a}\mathrm{d}x^{j}-\left(
P_{a}^{i}y_{,i}^{a}+P_{a}^{ij}y_{,ij}-L\right) \mathrm{d}x^{2}\right] .
\end{equation*}%
Let us decompose $P_{a}^{ij}$ into its symmetric and antisymmetric parts in
the upper indices%
\begin{equation}
P_{a}^{ij}=P_{a}^{(ij)}+P_{a}^{[ij]}.  \label{7.13a}
\end{equation}%
Then, 
\begin{equation}
P_{a}^{i}=\frac{\partial L}{\partial z_{i}^{a}}-P_{a,i_{1}}^{i_{1}i}=\frac{%
\partial L}{\partial z_{i}^{a}}%
-P_{a,i_{1}}^{(i_{1}i)}-P_{a,i_{1}}^{[i_{1}i]},  \label{7.13b}
\end{equation}%
so that 
\begin{eqnarray}
&&\int_{\Sigma }j^{\mathrm{3}}\sigma ^{\ast }\left( Y_{T}^{2}%
%TCIMACRO{%
%\TeXButton{lefthook}{{\mbox{$ \rule {5pt} {.5pt}\rule {.5pt} {6pt} \, $}}}}%
%BeginExpansion
{\mbox{$ \rule {5pt} {.5pt}\rule {.5pt} {6pt} \, $}}%
%EndExpansion
\Theta \right) =  \label{7.14a} \\
&=&\int_{\Sigma }\left[ P_{a}^{2}y_{,j}^{a}\mathrm{d}%
x^{j}+P_{a}^{2i_{2}}y_{,i_{2}j}^{a}\mathrm{d}x^{j}-\left(
P_{a}^{i}y_{,i}^{a}+P_{a}^{ij}y_{,ij}-L\right) \mathrm{d}x^{2}\right]  \notag
\\
&=&\int_{\Sigma }\left[ \left( \frac{\partial L}{\partial z_{2}^{a}}%
-P_{a,i_{1}}^{(i_{1}2)}-P_{a,i_{1}}^{[i_{1}2]}\right) y_{,j}^{a}\mathrm{d}%
x^{j}+\left( P_{a}^{(2i_{2})}+P_{a}^{[2i_{2}]}\right) y_{,i_{2}j}^{a}\mathrm{%
d}x^{j}\right]  \notag \\
&&-\int_{\Sigma }\left( \left( \frac{\partial L}{\partial z_{i}^{a}}%
-P_{a,i_{1}}^{(i_{1}i)}-P_{a,i_{1}}^{[i_{1}i]}\right)
y_{,i}^{a}+P_{a}^{ij}y_{,ij}-L\right) \mathrm{d}x^{2}  \notag \\
&=&\int_{\Sigma }\left[ \left( \frac{\partial L}{\partial z_{2}^{a}}%
-P_{a,i_{1}}^{(i_{1}2)}\right) y_{,j}^{a}\mathrm{d}x^{j}+\left(
P_{a}^{(2i_{2})}\right) y_{,i_{2}j}^{a}\mathrm{d}x^{j}-\right] +  \notag \\
&&-\int_{\Sigma }\left( \left( \frac{\partial L}{\partial z_{2}^{a}}%
-P_{a,i_{1}}^{(i_{1}2)}\right) y_{,i}^{a}+P_{a}^{ij}y_{,ij}-L\right) \mathrm{%
d}x^{2}+  \notag \\
&&+\int_{\Sigma }\left[ -P_{a,i_{1}}^{[i_{1}2]}y_{,j}^{a}\mathrm{d}%
x^{j}+P_{a}^{[2i_{2}]}y_{,i_{2}j}^{a}\mathrm{d}%
x^{j}+P_{a,i_{1}}^{[i_{1}i]}y_{,i}^{a}\mathrm{d}x^{2}\right] .  \notag
\end{eqnarray}%
The last integral in equations (\ref{7.14a}) involves only the odd terms $%
P_{a}^{[ij]}$. It can be rewritten as follows 
\begin{eqnarray}
&&\int_{\Sigma }\left[ -P_{a,i_{1}}^{[i_{1}2]}y_{,j}^{a}\mathrm{d}%
x^{j}+P_{a}^{[2i_{2}]}y_{,i_{2}j}^{a}\mathrm{d}%
x^{j}+P_{a,i_{1}}^{[i_{1}i]}y_{,i}^{a}\mathrm{d}x^{2}\right] =  \label{7.14b}
\\
&=&\int_{\Sigma }\left[ \left( -P_{a,1}^{[12]}-P_{a,2}^{[22]}\right)
y_{,j}^{a}\mathrm{d}x^{j}+\left(
P_{a}^{[21]}y_{,1j}^{a}+P_{a}^{[22]}y_{,2j}\right) \mathrm{d}x^{j}\right] 
\notag \\
&&+\int_{\Sigma }\left(
P_{a,1}^{[11]}y_{,1}^{a}+P_{a,1}^{[12]}y_{,2}^{a}+P_{a,2}^{[21]}y_{1}^{a}+P_{a,2}^{[22]}y_{,2}^{a}\right) 
\mathrm{d}x^{2}  \notag \\
&=&\int_{\Sigma }\left[ -P_{a,1}^{[12]}\left( y_{,1}^{a}\mathrm{d}%
x^{1}+y_{,2}^{a}\mathrm{d}x^{2}\right) +P_{a}^{[21]}y_{,1j}^{a}\mathrm{d}%
x^{j}+\left( P_{a,1}^{[12]}y_{,2}^{a}+P_{a,2}^{[21]}y_{1}^{a}\right) \mathrm{%
d}x^{2}\right]  \notag \\
&=&\int_{\Sigma }\left[ -P_{a,1}^{[12]}y_{,1}^{a}\mathrm{d}%
x^{1}+P_{a}^{[21]}y_{,1j}^{a}\mathrm{d}x^{j}+P_{a,2}^{[21]}y_{1}^{a}\mathrm{d%
}x^{2}\right]  \notag \\
&=&\int_{\Sigma }\left[ P_{a}^{[21]}y_{,1j}^{a}\mathrm{d}%
x^{j}-P_{a,1}^{[12]}y_{,1}^{a}\mathrm{d}x^{1}+P_{a,2}^{[21]}y_{1}^{a}\mathrm{%
d}x^{2}\right]  \notag \\
&=&\int_{\Sigma }\left[ \mathrm{d}\left( P_{a}^{[21]}y_{,1}^{a}\right)
-P_{a,j}^{[21]}y_{,1}^{a}\mathrm{d}x^{j}-P_{a,1}^{[12]}y_{,1}^{a}\mathrm{d}%
x^{1}+P_{a,2}^{[21]}y_{1}^{a}\mathrm{d}x^{2}\right]  \notag \\
&=&\int_{\Sigma }\left[ \mathrm{d}\left( P_{a}^{[21]}y_{,1}^{a}\right)
-P_{a,1}^{[21]}y_{,1}^{a}\mathrm{d}x^{1}-P_{a,2}^{[21]}y_{,1}^{a}\mathrm{d}%
x^{2}-P_{a,1}^{[12]}y_{,1}^{a}\mathrm{d}x^{1}+P_{a,2}^{[21]}y_{1}^{a}\mathrm{%
d}x^{2}\right]  \notag \\
&=&\int_{\Sigma }\left[ \mathrm{d}\left( P_{a}^{[21]}y_{,1}^{a}\right)
-\left( P_{a,1}^{[21]}y_{,1}^{a}\mathrm{d}x^{1}+P_{a,1}^{[12]}y_{,1}^{a}%
\mathrm{d}x^{1}\right) -P_{a,2}^{[21]}y_{,1}^{a}\mathrm{d}%
x^{2}+P_{a,2}^{[21]}y_{1}^{a}\mathrm{d}x^{2}\right]  \notag \\
&=&\int_{\Sigma }\mathrm{d}\left( P_{a}^{[21]}y_{,1}^{a}\right)  \notag
\end{eqnarray}

Since, we consider an evolution equation with non-compact Cauchy surfaces,
replace $K$ by a slice 
\begin{equation*}
S=\{(x^{1},x^{2})\in \mathbb{R}^{2}\mid 0<x^{1}<t\}
\end{equation*}%
with boundary 
\begin{equation*}
\partial S=\Sigma _{t}-\Sigma _{0}=\{(t,x_{2})\in \mathbb{R}%
^{2}\}-\{(0,x_{2})\in \mathbb{R}^{2}\}.
\end{equation*}%
We assume that the fields $y^{a}(x_{1},x_{2})$ vanish sufficiently fast as $%
x_{2}\rightarrow \pm \infty ,$so that integrals over $K,$ $\Sigma _{t}$ and $%
\Sigma _{0}$ converge and permit integration by \ parts. For $\Sigma =\Sigma
_{t}$, equations (\ref{7.14a}) \ and (\ref{7.14b}) yield 
\begin{eqnarray}
&&\int_{\Sigma _{t}}j^{\mathrm{3}}\sigma ^{\ast }\left( Y_{T}^{2}%
%TCIMACRO{%
%\TeXButton{lefthook}{{\mbox{$ \rule {5pt} {.5pt}\rule {.5pt} {6pt} \, $}}}}%
%BeginExpansion
{\mbox{$ \rule {5pt} {.5pt}\rule {.5pt} {6pt} \, $}}%
%EndExpansion
\Theta \right) =\int_{\Sigma _{t}}\left[ \left( \frac{\partial L}{\partial
z_{2}^{a}}-P_{a,i_{1}}^{(i_{1}2)}\right) y_{,j}^{a}\mathrm{d}x^{j}+\left(
P_{a}^{(2i_{2})}\right) y_{,i_{2}j}^{a}\mathrm{d}x^{j}\right]  \label{7.14c}
\\
&&-\int_{\Sigma _{t}}\left( \left( \frac{\partial L}{\partial z_{2}^{a}}%
-P_{a,i_{1}}^{(i_{1}2)}\right) y_{,i}^{a}+P_{a}^{ij}y_{,ij}-L\right) \mathrm{%
d}x^{2}+\int_{\Sigma _{t}}\mathrm{d}\left( P_{a}^{[21]}y_{,1}^{a}\right) 
\notag \\
&&=\int_{\Sigma _{t}}\left[ \left( \frac{\partial L}{\partial z_{2}^{a}}%
-P_{a,i_{1}}^{(i_{1}2)}\right) y_{,j}^{a}\mathrm{d}x^{j}+\left(
P_{a}^{(2i_{2})}\right) y_{,i_{2}j}^{a}\mathrm{d}x^{j}\right] +  \notag \\
&&-\int_{\Sigma _{t}}\left( \left( \frac{\partial L}{\partial z_{2}^{a}}%
-P_{a,i_{1}}^{(i_{1}2)}\right) y_{,i}^{a}+P_{a}^{ij}y_{,ij}-L\right) \mathrm{%
d}x^{2}+  \notag \\
&&+\lim_{x_{2}\rightarrow \infty }\left( P_{a}^{[21]}y_{,1}^{a}\right)
(t,x_{2})-\lim_{x_{2}\rightarrow -\infty }\left(
P_{a}^{[21]}y_{,1}^{a}\right) (t,x_{2})  \notag \\
&=&\int_{\Sigma _{t}}\left[ \left( \frac{\partial L}{\partial z_{2}^{a}}%
-P_{a,i_{1}}^{(i_{1}2)}\right) y_{,j}^{a}\mathrm{d}x^{j}+\left(
P_{a}^{(2i_{2})}\right) y_{,i_{2}j}^{a}\mathrm{d}x^{j}\right]  \notag \\
&&-\int_{\Sigma _{t}}\left( \left( \frac{\partial L}{\partial z_{2}^{a}}%
-P_{a,i_{1}}^{(i_{1}2)}\right) y_{,i}^{a}+P_{a}^{ij}y_{,ij}-L\right) \mathrm{%
d}x^{2}  \notag
\end{eqnarray}%
because our asymptotic conditions require that $\lim_{x_{2}\rightarrow
\infty }\left( P_{a}^{[21]}y_{,1}^{a}\right) (t,x_{2})=0$ and $%
\lim_{x_{2}\rightarrow -\infty }\left( P_{a}^{[21]}y_{,1}^{a}\right)
(t,x_{2})=0.$ Hence the potential non-uniquennes of constants of of motion
is taken care of by the appropriate choice of boundary conditions.\medskip

\section{Appendix}

\subsection{Jets}

Let $\pi :N\rightarrow M$ be a locally trivial fibration. A local section $%
\sigma $ of $\pi $ is a smooth map $\sigma :M\rightarrow N,$ defined on an
open subset $U$ of $M,$ such that $\pi \circ \sigma (x)=x$ for every $x\in U$%
. If $U=M$, we say that $\sigma :M\rightarrow N$ is a global section of $\pi 
$. In the following, we say $\sigma :M\rightarrow N$ is a section of $\pi $
if $\sigma $ is either local or global section.\medskip

Suppose that $m=\dim M$ and $m+n=\dim N$. We use local coordinates $(x^{i})$
on $M$, where $i=1,...,m$, and $(x^{i},y^{a})$ on $N,$ where $a=1,...,n.$
The local ccordinate description of a section $\sigma :M\rightarrow N$ is
given by $y^{a}=\sigma ^{a}(x^{1},...,x^{m})$ for $a=1,...n$.\medskip

For each $x\in M$ and $\mathrm{k}=1,2,...,$ sections $\sigma $ and $\check{%
\sigma}$ of $\pi $ are $\mathrm{k}$-equivalent at $x$ if $\sigma (x)=\check{%
\sigma}(x)$ and, in local coordinates, 
\begin{equation}
\sigma _{,i_{1}...i_{l}}^{a}(x)=\check{\sigma}_{,i_{1}...i_{l}}^{a}(x),
\label{jet}
\end{equation}%
where 
\begin{equation}
\sigma _{,i_{1}...i_{l}}^{a}(x)=\frac{{\small \partial }^{l}{\small \sigma }%
^{a}}{{\small \partial x}^{i_{1}}{\small ...\partial x}^{i_{l}}}%
(x^{1}(x),...,x^{m}(x)),  \label{jet1}
\end{equation}%
for all $l=1,...,\mathrm{k}.$\medskip

The $\mathrm{k}$-equivalence class at $x$ of a section $\sigma $ is called
the $\mathrm{k}$-jet of $\sigma $ at $x$ and denoted $j^{k\mathrm{k}}\sigma
(x)$. The space of $\mathrm{k}$-equivalence classes at $x$ of all section $%
\sigma $ is denoted $J_{x}^{\mathrm{k}}(M,N)$ and 
\begin{equation*}
J^{\mathrm{k}}(M,N)=\dbigcup\limits_{x\in M}J_{x}^{\mathrm{k}}(M,N)
\end{equation*}%
is called the space of $\mathrm{k}$-jets of sections of $\pi $. In terms of
local coordinates, $j^{\mathrm{k}}\sigma (x)$ has coordinates $%
(x^{i},y^{a},z_{i_{1}}^{a},....,z_{i_{1}...i_{k}}^{a}),$ where 
\begin{equation*}
z_{i_{1}...i_{l}}^{a}=\sigma _{,i_{1}...i_{l}}^{a}(x),
\end{equation*}%
for $l=1,...,k$, $i^{1},...,i^{l}=1,...,n$, and $a=1,...,m$. Since partial
derivatives of a smooth function commute, the variables $%
z_{i_{1}...i_{l}}^{a}$ cannot be considered as independent coordinates. In
the case when it matters, we use an independent collection 
\begin{equation}
\{z_{i_{1}...i_{l}}^{a}\mid a=1,...,m,\text{ and }1\leq i_{1}\leq i_{2}\leq
...\leq i_{l}\};  \label{jet coordinates}
\end{equation}%
see equation (\ref{Phi}). However, in general, we use symmetry of variables $%
z_{i_{1}...i_{l}}^{a}$ in the indices $i_{1},...,i_{l}.$ \medskip

There are several maps defined on $J^{\mathrm{k}}(M,N):$

\noindent \emph{the source map}%
\begin{equation*}
\pi ^{\mathrm{k}}:J^{\mathrm{k}}(M,N)\rightarrow M:j^{\mathrm{k}}\sigma
(x)\mapsto x,
\end{equation*}%
\noindent \emph{the target map}%
\begin{equation*}
\pi _{0}^{\mathrm{k}}:J^{\mathrm{k}}(M,N)\rightarrow N:j^{\mathrm{k}}\sigma
(x)\mapsto \sigma (x),
\end{equation*}%
\noindent \emph{the }$(\mathrm{k},\mathrm{l})$-\emph{forgetful }%
\begin{equation*}
\pi _{\mathrm{l}}^{\mathrm{k}}:J^{\mathrm{k}}(M,N)\rightarrow J^{\mathrm{l}%
}(M,N):j^{\mathrm{k}}\sigma (x)\mapsto j^{\mathrm{l}}\sigma (x)\text{ \ for }%
\mathrm{k}>\mathrm{l}>0.
\end{equation*}%
Each of these maps defines a fibre bundle structure in $J^{\mathrm{k}}(M,N)$%
. For this reason, $J^{\mathrm{k}}(M,N)$ is also called the $\mathrm{k}$-jet
bundle of sections of $\pi $.\medskip

Let $\sigma :M\rightarrow N$ be a section of $\pi :N\rightarrow M$%
. We denote the $\mathrm{k}$\emph{-jet extension} of $\sigma $ by 
\begin{equation*}
j^{\mathrm{k}}\sigma :M \rightarrow J^{\mathrm{k}}(M,N) :x\mapsto j^{\mathrm{k}%
}\sigma (x).
\end{equation*}%
For every integer $\mathrm{k}>0,$ 
\begin{equation}
\pi _{0}^{\mathrm{k}}\circ j^{\mathrm{k}}\sigma =\sigma .  \label{A1}
\end{equation}%
Similarly, for each $\mathrm{k}>\mathrm{l}>0$, 
\begin{equation*}
\pi _{\mathrm{l}}^{\mathrm{k}}\circ j^{\mathrm{k}}\sigma =j^{\mathrm{l}%
}\sigma .
\end{equation*}

A section $\rho :M\rightarrow J^{\mathrm{k}}(M,N)$ of the source map $\pi ^{%
\mathrm{k}}:J^{\mathrm{k}}(M,N)\rightarrow M$ is called \emph{holonomic} if
there exists a section $\sigma $ of $\pi $ such that 
\begin{equation*}
\rho =j^{\mathrm{k}}(\sigma ).
\end{equation*}%
It follows from equation (\ref{A1}) that $\rho $ is holonomic if and only if 
\begin{equation*}
\rho =j^{\mathrm{k}}(\pi _{0}^{\mathrm{k}}\circ \rho ).
\end{equation*}

Each local chart $M$, gives rise to \emph{local contact forms} on $%
J^{k}(M,N) $ given by 
\begin{eqnarray}
\vartheta ^{a} &=&\mathrm{d}y^{a}-\sum_{i=1}^{m}z_{i}^{a}\mathrm{d}%
x^{i},~~~\vartheta _{i}^{a}=\mathrm{d}z_{i}^{a}-\sum_{j=1}^{m}z_{ij}^{a}%
\mathrm{d}x^{j},~~~.....  \label{contact forms} \\
\vartheta _{i_{1}...i_{k-1}}^{a} &=&\mathrm{d}z_{i_{1}...i_{k-1}}^{a}-%
\sum_{i_{k}=1}^{m}z_{i_{1}...i_{k-1}i_{k}}^{a}\mathrm{d}x^{i_{k}}.  \notag
\end{eqnarray}

A section $\rho :M\rightarrow J^{\mathrm{k}}(M,N)$ of the source map $\pi ^{%
\mathrm{k}}$ is holonomic if the tangent space of its range is annihilated
by the contact forms $\vartheta _{i_{1}...i_{l}}^{a}$ for all $l=0,...,k$
and all indices $i_{1},...,i_{l}=1,...,n$ and any collection of coordinate
charts covering $M$.

\subsection{Prolongations}

Let $Y$ be a vector field on $N$ which projects to a vector field $Y^{0}$ on 
$M$. In other words, $Y$ is $\pi $-related to a vector field $Y^{0}$, that
is 
\begin{equation}
T\pi \circ Y=Y^{0}\circ \pi .  \label{A2}
\end{equation}%
This implies that $\pi :N\rightarrow M$ intertwines the actions of
local-one-parameter local groups $\mathrm{e}^{tY}$ and $\mathrm{e}^{tY^{0}}$
generated by $Y$ and $Y^{0}$, respectively.%
\begin{equation}
\begin{array}{ccccc}
&  & \mathrm{e}^{tY} &  &  \\ 
& N & \rightarrow & N &  \\ 
\pi & \downarrow &  & \downarrow & \pi \\ 
& M & \rightarrow & M &  \\ 
&  & \mathrm{e}^{tY^{0}} &  & 
\end{array}
\label{A3}
\end{equation}%
Hence, for every section $\sigma $ of $\pi $, 
\begin{equation}
\mathrm{e}^{tY\ast }\sigma =\mathrm{e}^{tY}\circ \sigma \circ \mathrm{e}%
^{-tY^{0}}  \label{A4}
\end{equation}%
is a local section of $\pi $. For every integer $\mathrm{k}$, the map $%
\sigma \mapsto \mathrm{e}^{tY\ast }\sigma $ induces a local-one parameter
local group%
\begin{equation}
\mathrm{e}^{tY^{\mathrm{k}}}:J^{\mathrm{k}}(M,N)\rightarrow J^{\mathrm{k}%
}(M,N):j^{\mathrm{k}}\sigma (x)\mapsto \lbrack j^{\mathrm{k}}(\mathrm{e}%
^{tY\ast }\sigma )](\mathrm{e}^{tY^{0}}x).  \label{A5}
\end{equation}%
of diffeomorphisms of $J^{\mathrm{k}}(M,N)$ to itself, generated by a vector
field $Y^{\mathrm{k}}$ on $J^{\mathrm{k}}(M,N),$ called the prolongation of $%
Y$ to $J^{\mathrm{k}}(M,N)$. In other words, 
\begin{equation}
\mathrm{e}^{tY^{\mathrm{k}}}\circ j^{\mathrm{k}}\sigma =j^{\mathrm{k}}(%
\mathrm{e}^{tY\ast }\sigma ).  \label{A6}
\end{equation}%
For every $0<\mathrm{l}<\mathrm{k}$, $Y^{\mathrm{k}}$ is $\pi _{\mathrm{l}}^{%
\mathrm{k}}$-related to $Y^{\mathrm{l}},$%
\begin{equation}
T\pi _{\mathrm{l}}^{\mathrm{k}}\circ Y^{\mathrm{k}}=Y^{\mathrm{l}}\circ \pi
_{\mathrm{l}}^{\mathrm{k}},  \label{A7}
\end{equation}%
where $\pi _{\mathrm{l}}^{\mathrm{k}}:J^{\mathrm{k}}(M,N)\rightarrow J^{%
\mathrm{l}}(M,N)$ is the forgetful map. \medskip

Following reference \cite{olver}, we show how to find the prolongation $Y^{%
\mathrm{k}}$ of a vector field $Y=Y^{i}(x)\frac{\partial }{\partial x^{i}}%
+Y^{a}(x,y)\frac{\partial }{\partial y^{a}}$ on $N$ that is $\pi $-related
to $Y^{0}=Y^{i}(x)\frac{\partial }{\partial x^{i}}$ on $M$ using the
condition that, for every local contact form $\vartheta $ on $J^{\mathrm{k}%
}(M,N)$, the Lie derivative $\pounds _{Y^{\mathrm{k}}}\vartheta $ of $%
\vartheta $ with respect to $Y^{\mathrm{k}}$ is a linear combination of
local contact forms. Let 
\begin{equation*}
Y^{\mathrm{k}}=Y^{i}\frac{\partial }{\partial x^{i}}+Y^{a}\frac{\partial }{%
\partial y^{a}}+Y_{i}^{a}\frac{\partial }{\partial z_{i}^{a}}%
+...+Y_{i_{1}i_{2}...i_{k}}^{a}\frac{\partial }{\partial
z_{i_{1}i_{2}...i_{k}}^{a}}
\end{equation*}%
be the prolongation of $Y$ to $J^{\mathrm{k}}(M,N).$ Then, 
\begin{eqnarray*}
&&\pounds _{Y^{\mathrm{k}}}[\mathrm{d}y^{a}-z_{i}^{a}\mathrm{d}x^{i}] \\
&=&Y^{\mathrm{k}}%
%TCIMACRO{%
%\TeXButton{lefthook}{{\mbox{$ \rule {5pt} {.5pt}\rule {.5pt} {6pt} \, $}}}}%
%BeginExpansion
{\mbox{$ \rule {5pt} {.5pt}\rule {.5pt} {6pt} \, $}}%
%EndExpansion
\left( \mathrm{d}[\mathrm{d}y^{a}-z_{i}^{a}\mathrm{d}x^{i}]\right) +\mathrm{d%
}\left( Y^{\mathrm{k}}%
%TCIMACRO{%
%\TeXButton{lefthook}{{\mbox{$ \rule {5pt} {.5pt}\rule {.5pt} {6pt} \, $}}}}%
%BeginExpansion
{\mbox{$ \rule {5pt} {.5pt}\rule {.5pt} {6pt} \, $}}%
%EndExpansion
[\mathrm{d}y^{a}-z_{i}^{a}\mathrm{d}x^{i}]\right) \\
&=&-Y^{\mathrm{k}}%
%TCIMACRO{%
%\TeXButton{lefthook}{{\mbox{$ \rule {5pt} {.5pt}\rule {.5pt} {6pt} \, $}}}}%
%BeginExpansion
{\mbox{$ \rule {5pt} {.5pt}\rule {.5pt} {6pt} \, $}}%
%EndExpansion
\left( \mathrm{d}z_{i}^{a}\wedge \mathrm{d}x^{i}\right) +\mathrm{d}\left(
Y^{a}-z_{i}^{a}Y^{i}\right) \\
&=&-Y_{i}^{a}\mathrm{d}x^{i}+Y^{i}\mathrm{d}z_{i}^{a}+Y_{,b}^{a}\mathrm{d}%
y^{b}+Y_{,i}^{a}\mathrm{d}x^{i}-Y^{i}\mathrm{d}z_{i}^{a}-z_{i}^{a}Y_{,j}^{i}%
\mathrm{d}x^{j} \\
&=&Y_{,b}^{a}(\mathrm{d}y^{b}-z_{i}^{b}\mathrm{d}x^{i})+Y_{,b}^{a}z_{i}^{b}%
\mathrm{d}x^{i}-Y_{i}^{a}\mathrm{d}x^{i}+Y_{,i}^{a}\mathrm{d}%
x^{i}-z_{i}^{a}Y_{,j}^{i}\mathrm{d}x^{j} \\
&=&Y_{,b}^{a}(\mathrm{d}y^{b}-z_{i}^{b}\mathrm{d}x^{i})+\left[ \left(
Y_{,b}^{a}z_{i}^{b}-z_{j}^{a}Y_{,i}^{j}+Y_{,i}^{a}\right) -Y_{i}^{a}\right] 
\mathrm{d}x^{i},
\end{eqnarray*}%
which implies that 
\begin{equation}
Y_{i}^{a}=Y_{,b}^{a}z_{i}^{b}-z_{j}^{a}Y_{,i}^{j} +Y_{,i}^{a}.  \label{A8}
\end{equation}%
Similarly, 
\begin{eqnarray*}
&&\pounds _{Y^{\mathrm{k}}}[\mathrm{d}z_{i}^{a}-z_{ij}^{a}\mathrm{d}x^{i}] \\
&=&Y^{\mathrm{k}}%
%TCIMACRO{%
%\TeXButton{lefthook}{{\mbox{$ \rule {5pt} {.5pt}\rule {.5pt} {6pt} \, $}}}}%
%BeginExpansion
{\mbox{$ \rule {5pt} {.5pt}\rule {.5pt} {6pt} \, $}}%
%EndExpansion
\left( \mathrm{d}[\mathrm{d}z_{i}^{a}-z_{ij}^{a}\mathrm{d}x^{i}]\right) +%
\mathrm{d}\left( Y^{k}%
%TCIMACRO{%
%\TeXButton{lefthook}{{\mbox{$ \rule {5pt} {.5pt}\rule {.5pt} {6pt} \, $}}}}%
%BeginExpansion
{\mbox{$ \rule {5pt} {.5pt}\rule {.5pt} {6pt} \, $}}%
%EndExpansion
[\mathrm{d}z_{i}^{a}-z_{ij}^{a}\mathrm{d}x^{i}]\right) \\
&=&-Y^{\mathrm{k}}%
%TCIMACRO{%
%\TeXButton{lefthook}{{\mbox{$ \rule {5pt} {.5pt}\rule {.5pt} {6pt} \, $}}}}%
%BeginExpansion
{\mbox{$ \rule {5pt} {.5pt}\rule {.5pt} {6pt} \, $}}%
%EndExpansion
\left( \mathrm{d}z_{ij}^{a}\wedge \mathrm{d}x^{j}\right) +\mathrm{d}\left(
Y_{i}^{a}-z_{ij}^{a}Y^{i}\right) \\
&=&-Y_{ij}^{a}\mathrm{d}x^{j}+Y^{j}\mathrm{d}z_{ij}^{a}+Y_{i,b}^{a}\mathrm{d}%
y^{b}+Y_{i,j}^{a}\mathrm{d}x^{j}-Y^{i}\mathrm{d}%
z_{ij}^{a}-z_{ij}^{a}Y_{,k}^{i}\mathrm{d}x^{k} \\
&=&Y_{i,b}^{a}\left( \mathrm{d}y^{b}-z_{j}^{b}\mathrm{d}x^{j}\right)
+Y_{i,b}^{a}z_{j}^{b}\mathrm{d}x^{j}-Y_{ij}^{a}\mathrm{d}x^{j}++Y_{i,j}^{a}%
\mathrm{d}x^{j}-z_{ij}^{a}Y_{,k}^{i}\mathrm{d}x^{k} \\
&=&Y_{i,b}^{a}\left( \mathrm{d}y^{b}-z_{j}^{b}\mathrm{d}x^{j}\right) +\left[
\left( Y_{i,b}^{a}z_{j}^{b}-z_{ki}^{a}Y_{,j}^{k}+Y_{i,j}^{a}\right)
-Y_{ij}^{a}\right] \mathrm{d}x^{j},
\end{eqnarray*}%
so that, for $i\leq j$,%
\begin{equation*}
Y_{ij}^{a}=Y_{i,b}^{a}z_{j}^{b}-z_{ki}^{a}Y_{,i}^{k}+Y_{i,j}^{a}.
\end{equation*}%
Symmetrizing, we get 
\begin{equation}
Y_{ij}^{a}=z_{(j}^{b}Y_{i),b}^{a}-z_{k(i}^{a}Y_{,i)}^{k}+Y_{(i,j)}^{a}.
\label{A9}
\end{equation}

In general,%
\begin{eqnarray*}
&&\pounds _{Y^{\mathrm{k}}}[\mathrm{d}%
z_{i_{1}...i_{l}}^{a}-z_{i_{1}...i_{l}j}^{a}\mathrm{d}x^{j}]= \\
&=&Y^{\mathrm{k}}%
%TCIMACRO{%
%\TeXButton{lefthook}{{\mbox{$ \rule {5pt} {.5pt}\rule {.5pt} {6pt} \, $}}}}%
%BeginExpansion
{\mbox{$ \rule {5pt} {.5pt}\rule {.5pt} {6pt} \, $}}%
%EndExpansion
\left( \mathrm{d}[\mathrm{d}z_{i_{1}...i_{l}}^{a}-z_{i_{1}...i_{l}j}^{a}%
\mathrm{d}x^{j}]\right) +\mathrm{d}\left( Y^{\mathrm{k}}%
%TCIMACRO{%
%\TeXButton{lefthook}{{\mbox{$ \rule {5pt} {.5pt}\rule {.5pt} {6pt} \, $}}}}%
%BeginExpansion
{\mbox{$ \rule {5pt} {.5pt}\rule {.5pt} {6pt} \, $}}%
%EndExpansion
\right) \\
&=&-Y^{\mathrm{k}}%
%TCIMACRO{%
%\TeXButton{lefthook}{{\mbox{$ \rule {5pt} {.5pt}\rule {.5pt} {6pt} \, $}}}}%
%BeginExpansion
{\mbox{$ \rule {5pt} {.5pt}\rule {.5pt} {6pt} \, $}}%
%EndExpansion
\left( \mathrm{d}z_{i_{1}...i_{l}j}^{a}\wedge \mathrm{d}x^{j}\right) +%
\mathrm{d}\left( Y_{i_{1}...i_{l}}^{a}-z_{i_{1}...i_{l}j}^{a}Y^{j}\right) \\
&=&-Y_{i_{1}...i_{l}j}^{a}\mathrm{d}x^{j}+Y^{j}\mathrm{d}%
z_{i_{1}...i_{l}j}^{a}+Y_{i_{1}...i_{l},b}^{a}\mathrm{d}%
y^{b}+Y_{i_{1}...i_{l},j}^{a}\mathrm{d}x^{j} \\
&&-Y^{j}\mathrm{d}z_{i_{1}...i_{l}j}^{a}-z_{i_{1}...i_{l}j}^{a}Y_{,k}^{j}%
\mathrm{d}x^{k} \\
&=&-Y_{i_{1}...i_{l}j}^{a}\mathrm{d}x^{j}+Y_{i_{1}...i_{l},b}^{a}\mathrm{d}%
y^{b}+Y_{i_{1}...i_{l},j}^{a}\mathrm{d}x^{j}-z_{i_{1}...i_{l}j}^{a}Y_{,k}^{j}%
\mathrm{d}x^{k} \\
&=&Y_{i_{1}...i_{l},b}^{a}\left( \mathrm{d}y^{b}-z_{j}^{b}\mathrm{d}%
x^{j}\right) +Y_{i_{1}...i_{l},b}^{a}z_{j}^{b}\mathrm{d}%
x^{j}-Y_{i_{1}...i_{l}j}^{a}\mathrm{d}x^{j} \\
&&+Y_{i_{1}...i_{l},j}^{a}\mathrm{d}x^{j}-z_{i_{1}...i_{l}j}^{a}Y_{,k}^{j}%
\mathrm{d}x^{k} \\
&=&Y_{i_{1}...i_{l},b}^{a}\left( \mathrm{d}y^{b}-z_{j}^{b}\mathrm{d}%
x^{j}\right) +\left[ \left(
Y_{i_{1}...i_{l},b}^{a}z_{j}^{b}-z_{i_{1}...i_{l}k}^{a}Y_{,j}^{k}+Y_{i_{1}...i_{l},j}^{a}\right) -Y_{i_{1}...i_{l},j}^{a}%
\right] \mathrm{d}x^{j}.
\end{eqnarray*}%
Therefore 
\begin{equation}
Y_{i_{1}i_{2...}i_{l}j}^{a}=z_{(j}^{b}Y_{i_{1}...i_{l}),b}^{a}-z_{k(i_{1}i_{2}...i_{l}}^{a}Y_{,j)}^{k}+Y_{(i_{1}...i_{l},j)}^{a}.
\label{A 10}
\end{equation}

\bigskip

Department of Mathematics and Statistics,

University of Calgary,

Calgary, Alberta, Canada.

sniatycki@gmail.com

\bigskip

Department of Mechanical Engineering

Ben Gurion University of the Negev,

Beer-Sheeva, Israel.

rsegev@bgu.ac.il

\end{document}